\documentclass{LMCS}

\def\dOi{10(3:21)2014}
\lmcsheading%
{\dOi}
{1--43}
{}
{}
{Dec.~\phantom05, 2013}
{Sep.~12, 2014}
{}

\ACMCCS{[{\bf Theory of computation}]: Semantics and
  reasoning---Program reasoning---Program analysis}

\usepackage{mathtools,amssymb}
\usepackage[
  pdfauthor={Matija Pretnar},
  pdftitle={Inferring Algebraic Effects}
]{hyperref}
\usepackage{listings}
\usepackage{mathpartir}
\usepackage{booktabs}
\usepackage[only, llbracket, rrbracket]{stmaryrd}
\usepackage{xspace,hyperref}
\input{macros}
\lstnewenvironment{source}{\lstset{
basicstyle=\ttfamily\small,
}}{}
\let\inline\lstinline

\begin{document}

\title{Inferring Algebraic Effects}

\author{Matija Pretnar}
\address{Faculty of Mathematics and Physics, University of Ljubljana, Slovenia}
\email{matija@pretnar.info}


\begin{abstract}
  \noindent
  We present a complete polymorphic effect inference algorithm
  for an ML-style language with handlers of not only exceptions,
  but of any other algebraic effect such as
  input \& output, mutable references and many others.

  Our main aim is to offer the programmer a useful insight into the effectful behaviour of programs.
  Handlers help here by cutting down possible effects
  and the resulting lengthy output that often plagues precise effect systems.
  Additionally, we present a set of methods
  that further simplify the displayed types,
  some even by deliberately hiding inferred information from the programmer.
\end{abstract}

\maketitle

\noindent
Though Haskell~\cite{jones2003haskell} fans may not think it is better to write impure programs in ML~\cite{milner1997definition},
they do agree it is easier.
You can insert a harmless printout without rewriting the rest of the program,
and you can combine multiple effects without a monad transformer.
This flexibility comes at a cost, though ---
ML types offer no insight into what effects may happen.
The suggested solution is to use an effect system~\cite{lucassen1988polymorphic,talpin1992type,benton1998compiling,tolmach1998optimizing,wadler1999marriage,benton2002monads,rytz2012lightweight},
which enriches existing types with information about effects.

An effect system can play two roles:
it can be \emph{descriptive} and inform about potential effects,
and it can be \emph{prescriptive} and limit the allowed ones.
In this paper, we focus on the former.
It turns out that striking a balance between expressiveness and simplicity of a descriptive effect system is hard.
One of the bigger problems is that effects tend to pile up,
and if the effect system takes them all into account,
we are often left with a lengthy output listing every single effect there is.

In this paper, we present a complete inference algorithm
for an expressive and simple descriptive polymorphic effect system of \Eff~\cite{bauer2012programming}
(freely available at \url{http://eff-lang.org}),
an ML-style language with handlers of not only exceptions, but of any other \emph{algebraic effect}~\cite{plotkin2003algebraic}
such as input \& output, non-determinism, mutable references and many others~\cite{plotkin2009handlers, bauer2012programming}.
Handlers prove to be extremely versatile and can express
  stream redirection,
  transactional memory,
  backtracking,
  cooperative multi-threading,
  delimited continuations,
and, like monads, give programmers a way to define their own.
And as handlers eliminate effects, they make the effect system \emph{non-monotone},
which helps with the above issue of a snowballing output.

We start in Section~\ref{sec:eff} with an tour of \Eff in which we informally explain handlers and show the main features of the proposed effect system.
Afterwards, we switch to formal development and in Section~\ref{sec:core-eff},
we recap the effect system for \emph{core Eff}~\cite{bauer2013effect},
a minimal formalization of \Eff.
Our contributions are:
\begin{itemize}
\item
  A set of syntax-directed rules for inferring types and constraints they must satisfy (Section~\ref{sec:inferring}).
  We show that these rules are sound and complete with regard to the effect system.
\item
  A unification algorithm that decomposes a set of constraints to a more basic form
  and decides if it admits a solution (Section~\ref{sec:unifying}).
  Unification fails only in the case of type mismatch,
  which fits our goal of a descriptive effect system
  that just refines existing ML types with details about effects.
\item
  A number of techniques that reduce the number of constraints without changing the set of solutions (Section~\ref{sec:simplifying}).
  The heavy lifting is done by \emph{garbage collection} of constraints~\cite{pottier2001simplifying, simonet2003type, trifonov1996subtyping},
  which we borrow and adapt slightly to our purpose.
  We also introduce a few more techniques particular to the algebraic setting.
\item
  A further collection of tactics for simplifying the display of inferred types (Section~\ref{sec:displaying}).
  To fully achieve their purpose,
  some of the presented tactics deliberately hide information from the programmer,
  though entire information is always used in the background.
\end{itemize}
We conclude by showing a couple of full runs of the algorithm (Section~\ref{sec:examples}) and by discussing related and future work.
To preserve the flow of the paper, we gather the (mostly routine) proofs in Appendix~\ref{app:proofs}.

\section{\Eff}
\label{sec:eff}

Effects aside, \Eff should be familiar to anyone that has worked with OCaml~\cite{ocaml}.
For example, the \inline{map} function,
which applies a function \inline{f} to each element of the list \inline{xs}
is defined as:
\begin{source}
  let rec map f xs =
    match xs with
    | [] -> []
    | x :: xs -> f x :: map f xs
\end{source}
The first important feature that distinguishes \Eff from OCaml is its effect system.
For example, the inferred type of \inline{map} is 
\[
  \kord{map} \T (\alpha \to[\drt] \beta)
  \to (\alpha \listty \to[\drt] \beta \listty)
\]
Here, the annotation on the arrow, called the \emph{dirt},
describes any effects that the function may trigger.
This additional information is very easy to understand:
the function \inline{map f} causes exactly the same effects~$\drt$ as \inline{f}.
The lack of dirt on the second arrow signifies that the application of \inline{map} to \inline{f} is pure.
Inferred types of some other typical higher-order functions are:
\begin{align*}
  \kord{compose} &\T (\alpha \to[\drt] \beta) \to (\beta \to[\drt'] \gamma) \to (\alpha \to[\drt \cup \drt'] \gamma) \\
  \kord{curry} &\T (\alpha \times \beta \to[\drt] \gamma) \to (\alpha \to \beta \to[\drt] \gamma) \\
  \kord{uncurry} &\T (\alpha \to[\drt] \beta \to[\drt'] \gamma) \to (\alpha \times \beta \to[\drt \cup \drt'] \gamma) \\
  \kord{fold\_left} &\T (\alpha \to[\drt] \beta \to[\drt'] \alpha) \to \alpha \to \beta \listty \to[\drt \cup \drt'] \alpha \\
  \kord{fold\_right} &\T (\alpha \to[\drt] \beta \to[\drt'] \beta) \to \alpha \listty \to \beta \to[\drt \cup \drt'] \beta \\
  \kord{filter} &\T (\alpha \to[\drt] \boolty) \to (\alpha \listty \to[\drt] \alpha \listty)
\end{align*}
The dirt annotations are similar to ones usually given in existing polymorphic effect systems such as~\cite{leroy2000type}.
Also note that if we disable the display of annotations, \Eff shows the programmer exactly the same types as OCaml.

The second distinguishing feature is that \Eff is based on algebraic effects~\cite{plotkin2001adequacy, plotkin2003algebraic}.
This means that effects are accessed exclusively through a set of \emph{operations},
which are all of the form $\hash{\inst}{\op}$,
  where an \emph{operation symbol}~$\op$ describes the action,
  and an \emph{instance}~$\inst$ describes where it should happen.
This means that we print to the standard output channel
using an operation \inline{std#print} instead of a primitive function \inline{print_string},
we access a reference \inline{r}
through operations \inline{r#lookup} and \inline{r#update} instead of through \inline{!} and \inline{:=},
and we raise an exception \inline{exc}
by calling \inline{exc#raise} instead of using a special keyword \inline{raise}.
The latter constructs are, of course, all definable in terms of the former.

This uniform representation allows effects to be seamlessly combined~\cite{hyland2006combining},
forms a natural basis for an effect system ---
possible effects of a given computation are captured accurately by the set of operations it calls ---
and paves the way for the third distinguishing feature of \Eff: handlers of arbitrary algebraic effects~\cite{plotkin2009handlers}.

Let us take a look at a few concrete examples.
What follows is by no means an exhaustive list of what can be accomplished with handlers.
For more examples, please see~\cite{plotkin2009handlers,bauer2012programming,kammar2013handlers}.

\subsection{Exceptions}

We start with exceptions as they are the simplest algebraic effect and
as exception handlers are already a well established concept.

\subsubsection{Effect types}

Instances are first-class values in \Eff and are given an \emph{effect type}.
For exception instances, called simply \emph{exceptions},
the effect type is of the form $\alpha \kpost{exception}$,
where $\alpha$ is the type of any additional information that the exception may carry.

We declare each effect type together with an \emph{effect signature} that lists
available operation symbols together with the types of parameters they accept
and of results they yield to the waiting continuation.
The signature corresponding to $\alpha \kpost{exception}$ is
\[
  \set{
    \kord{raise} \T \alpha \to \emptyty
  }
\]
where the result type of \inline{raise} is $\emptyty$ (a sum type with zero constructors)
because raising an exception terminates the computation and yields nothing to the continuation.

Each term of an effect type is also annotated with a \emph{region},
which describes (an over-approximation of) the set of instances it may occupy.
For example, an exception \inline{exc1} is given
the type $\alpha \kpost{exception}^{\set{\kord{exc1}}}$,
while the conditional \inline{if cond then exc1 else exc2} is given
the type $\alpha \kpost{exception}^{\set{\kord{exc1}, \kord{exc2}}}$
(any bigger sets would also be fine).

\subsubsection{Raising exceptions}
\label{ssub:raising-exceptions}

Since there is little we can do with a result of the empty type,
we rarely raise exceptions directly with \inline{exc#raise}.
Instead, we define a convenience function, also named \inline{raise},
which places the exception call inside the empty type eliminator:
\begin{source}
  let raise exc arg =
    match (exc#raise arg) with
\end{source}
with the inferred type
\[
  \kord{raise} \T \alpha \kpost{exception}^{\rgn} \to \alpha \to[\kord{raise} \T \rgn] \beta
\]
So, \inline{raise} takes an exception from a region captured by a region parameter~$\rgn$
and a matching argument of type $\alpha$.
The second application results in a computation that can raise any exception from $\rgn$,
while its return type~$\beta$ can be any arbitrary.
This allows us to use \inline{raise} at any point in the computation,
for example in computing the tail of a list:
\begin{source}
  let tail xs =
    match xs with
    | [] -> raise emptyListTail
    | _ :: xs -> xs
\end{source}
where \inline{emptyListTail} is the exception stating that the empty list has no tail.
The inferred type of \inline{tail} is
\[
  \kord{tail} \T \alpha \listty \to[\kord{raise} \T \set{\kord{emptyListTail}}] \alpha \listty
\]
where the dirt shows that \inline{emptyListTail} may be raised during execution.

In addition to inferring types of values, we can also infer \emph{dirty types~$A \E \Drt$} of computations.
These consist of a value type~$A$ describing the result
and of a dirt~$\Drt$ describing the operations that may be called while computing it.
For example, the inferred dirty type of the computation \inline{tail [1; 2; 3]} is $\type{int} \listty \E \set{\kord{raise} \T \set{\kord{emptyListTail}}}$.
Note that the effect system is conservative and signals a possible exception even though it will not be raised during runtime.
This also means that a computation is guaranteed to be pure whenever its inferred dirt is empty.

For a final example of the effect system before we turn to handlers,
let us combine \inline{map} and \inline{tail} as
\begin{source}
  let map_tail f xs = map f (tail xs)
\end{source}
with the inferred type
\[
  \kord{map\_tail} \T (\alpha \to[\kord{raise} \T \rgn \mid \drt] \beta) \to
  (\alpha \listty \to[\kord{raise} \T  \set{\kord{emptyListTail}} \cup \rgn \, \mid \, \drt] \beta \listty)
\]
Unlike the argument to \inline{map},
the argument of \inline{map_tail} has its dirt split into two parts:
the region parameter $\rgn$ captures the region of all the exceptions that \inline{f} may raise,
while the dirt parameter $\drt$ captures all other operations.
The split allows us to express the fact that \inline{map_tail f} may raise either \inline{emptyListTail} or an exception from $\rgn$,
and that it may call any operation from $\drt$ that \inline{f} can.
This is similar to \emph{row typing}~\cite{remy1993type},
where we use a parameter to capture all the fields of a record that are not explicitly mentioned.

\subsubsection{Handling exceptions}
\label{ssub:handling-exceptions}

As exceptions are the simplest example of effects,
exception handlers are the simplest example of handlers.
A \inline{try with} handling construct known from OCaml, like the one in
\begin{source}
  let print_ratio x y =
    try
      print_string ("Ratio is " ^ string_of_float (x /. y))
    with
    | Division_by_zero -> print_string "Ratio does not exist!"
\end{source}
would be used in \Eff as
\begin{source}
  let print_ratio x y =
    handle
      std#print ("Ratio is " ^ string_of_float (x /. y))
    with
    | divisionByZero#raise _ _ -> std#print "Ratio does not exist!"
\end{source}
Each handler case takes two arguments which we ignore for now.

Like instances, handlers are first-class values in \Eff,
and the above is just special syntax for
\begin{source}
  let print_ratio x y =
    with h handle
      std#print ("Ratio is " ^ string_of_float (x /. y))
\end{source}
where \inline{h} is given by
\begin{source}
  handler
  | divisionByZero#raise _ _ -> std#print "Ratio does not exist!"
\end{source}
Handlers are given \emph{handler types~$A \E \Drt \hto B \E \Drt'$},
meaning that a handler takes a computation with \emph{incoming} dirty type $A \E \Drt$
and transforms it into a computation with \emph{outgoing} dirty type $B \E \Drt'$.
The inferred type for \inline{h} is
\begin{multline*}
  \unitty \E \set{
    \kord{raise} \T \rgn_1,
    \kord{print} \T \rgn_2
  \mid \drt} \ \hto \\
  \unitty \E \set{
    \kord{raise} \T \rgn_1 - \kord{divisionByZero},
    \kord{print} \T \set{\kord{std}} \cup \rgn_2
  \mid \drt}
\end{multline*}
We see that \inline{h} leaves the type of results to be $\unitty$.
Its incoming dirt is split into three parts:
exceptions $\rgn_1$ that we may raise,
channels $\rgn_2$ that we may print to,
and any other operations~$\drt$ different from \inline{raise} or \inline{print} that we may call.
The outgoing dirt is similarly split:
the handled computation may still raise exceptions from $\rgn_1$,
except that now \inline{divisionByZero} will be handled.
Next, the handled computation may print to \inline{std} in addition to any channel in $\rgn_2$.
Finally, any operation in $\drt$ will be neither caught nor called by \inline{h},
so this part remains as it is.

Handler types usually (but not always) all have the same shape:
they remove certain operations,
possibly add some of their own,
and pass through any unhandled dirt.
This results in a repetitive type, which can be written in a more compact form
that emphasises only the differences.
For \inline{h}, this is:
\[
  \unitty \hto[\kord{raise} \T -\kord{divisionByZero},\ \kord{print} \T +\kord{std}] \unitty
\]
So, \inline{h} removes \inline{divisionByZero#raise}, adds \inline{std#print}, and leaves the rest as it is.

In practice, exception handlers are rarely reused,
because treatment of exceptions varies greatly
depending on the context in which they are handled.
An example of a more general exception handler is \inline{optionalize},
which transforms a computation into one that yields an optional result,
depending on if a given exception~\inline{exc} was raised or not.
We define \inline{optionalize} as:
\begin{source}
  let optionalize exc =
    handler
    | exc#raise _ _ -> None
    | val x -> Some x
\end{source}
The first thing to notice is the special \inline{val} case,
which determines what to do when the handled computation returns a value.
In our case, we wrap it with the \inline{Some} constructor of the $\kord{option}$ type.
Then, if we define
\begin{source}
  let tail_opt xs =
    with (optionalize emptyListTail) handle (tail xs)
\end{source}
the call \inline{tail_opt [1; 2; 3]} evaluates to \inline{Some [2; 3]},
while \inline{tail_opt []} evaluates to \inline{None}.
The inferred type of \inline{tail_opt} is the pure $\alpha \kpost{list} \to (\alpha \kpost{list}) \kpost{option}$,
while the type of \inline{optionalize} is
\[
  \kord{optionalize} \T \alpha \kpost{exception}^{\rgn} \to (\beta \hto[\kord{raise} \T \dotminus \rgn] \beta \kpost{option})
\]
So, we change the result type from $\beta$ to $\beta \kpost{option}$,
and remove any call of \inline{exc#raise} for the exception \inline{exc} determined by the region~$\rgn$.

\subsubsection{Handled regions}
\label{ssub:handled-regions}

You may have observed that we wrote $\dotminus \rgn$ instead of $-\rgn$ in the type of \inline{optionalize}.
This is meant to point out that we may remove $\rgn$ from the dirt only if it denotes a singleton.
The reason for this is as follows.

Take exceptions \inline{exc1} and \inline{exc2} and define a handler \inline{h} as:
\begin{source}
  let exc = if cond then exc1 else exc2 in
  let h = 
    handler
    | exc#raise _ _ -> ...
\end{source}
So, does the \inline{exc#raise} case of \inline{h} handle \inline{exc1#raise} or \inline{exc2#raise} at runtime?
Unfortunately, just by looking at the type
$\kord{exc} \T \alpha \kpost{exception}^{\set{\kord{exc1}, \kord{exc2}}}$,
we cannot say anything,
so we must assume that both are unhandled.
The only way we can be sure during type-checking that a handling case removes a given operation is
when the region of its instance is a singleton.

For this reason, we write $\dotminus \rgn$ whenever the region of the handled instance is still some parameter~$\rgn$.
If this eventually turns out to denote some singleton $\set{\inst}$, we may safely replace it by $-\inst$.
But if it turns to be a bigger set, we need to drop it from the handler type.

\subsection{Input \& output}

Interactive input \& output is also a very simple algebraic effect,
yet its handlers expose almost all the important aspects of general ones.
For input \& output, the effect instances are called \emph{channels}
and have the effect type $\kord{channel}$ with the signature
\[
  \set{
    \kord{read} \T \unitty \to \stringty,
    \kord{print} \T \stringty \to \unitty
  }
\]
A simple output handler is one that reverses the order of printouts:
\begin{source}
  handler
  | std#print msg k -> k (); std#print msg
\end{source}
The \inline{std#print} case takes two parameters:
  the string \inline{msg} given to \inline{std#print} and
  the continuation \inline{k} waiting for its result.
We first resume the continuation by passing it the unit value \inline{()},
and only after that finishes, we print out \inline{msg}.
The handler recursively handles any other \inline{std#print} that the continuation may call,
so the order of printouts in the continuation is reversed as well.
Note, however, that \inline{std#print} on the right-hand side is outside the scope of the handler
and remains unhandled (unless there are more handlers nested outside).

A more interesting example that can be useful in unit testing is a handler that
collects all printouts to a list of strings and returns it together with the result:
\begin{source}
  handler
  | val x -> (x, [])
  | std#print msg k ->
      let (x, msgs) = k () in
      (x, msg :: msgs)
\end{source}
So, a computation that just returns the value \inline{x} prints nothing,
and we return an empty list \inline{[]} together with \inline{x}.
If, however, the computation prints the string \inline{msg},
we resume the continuation \inline{k}.
This is also handled with the same handler, so it returns some value \inline{x}
and a list of its messages \inline{msgs}.
Now, we only need to prepend \inline{msg} to this list and return it together with \inline{x}.
The fact that the handler changes the type of the handled computation is reflected in its inferred type
$\alpha \hto[\kord{print} \T -\kord{std}] \alpha \times \stringty \kpost{list}$.

A matching handler that is also useful in unit testing is one that feeds a list of strings to \inline{std#read}:
\begin{source}
  handler
  | val x -> (fun strs -> x)
  | std#read () k -> (fun strs ->
      match strs with
      | str :: strs' -> (k str) strs'
      | [] -> (k "") []
    )
\end{source}
We accomplish this by transforming a computation into a function that accepts a list of strings \inline{strs}.
If the computation returns some value \inline{x},
the function ignores its argument and returns \inline{x}.
But if the computation calls \inline{std#read}, we take a look at the list \inline{strs}.
If it is non-empty, we pass \inline{k} its first element \inline{str}.
And since the continuation is further handled,
it is also a function that accepts a list of strings,
so we pass it the remainder \inline{strs'}.
If it is empty, we pass \inline{k} the empty string and again the empty list.

The inferred type of the handler is
\[
  \alpha \E \set{\kord{read} \T \rgn \mid \drt} \ \hto\  (\stringty \listty \to[\Drt] \alpha) \E \Drt
\]
where $\Drt = \set{\kord{read} \T \rgn - \kord{std} \mid \drt}$.
The reason $\Drt$ appears in two places is
because operations other than \inline{std#read} may occur before or after we handle the first \inline{std#read} and so obtain a function.
Since $\Drt$ appears twice, we unfortunately cannot use the compact form.

\label{page:finally}
However, \Eff extends handlers with an additional \inline{finally} case,
which first transforms the computation with the handler,
and then applies the finally case computation to the resulting value.
In particular, if \inline{h} is some handler and \inline{h_fin} is the same as \inline{h}
except that it also contains a case \inline{finally x -> c_fin},
the computation \inline{with h_fin handle c} behaves exactly as
\begin{source}
  let x = (with h handle c) in c_fin
\end{source}
Using this extension, we can define
\begin{source}
  let supply_input strs0 =
    handler
    | val x -> (* ...as before... *)
    | std#read () k -> (* ...as before... *)
    | finally f -> f strs0
\end{source}
So, once we get back a function \inline{f} accepting a list of inputs, we apply it to the given list \inline{strs0}.
We use the handler as
\begin{source}
  with
    supply_input ["Alpha"; "Bravo"; "Charlie"]
  handle
    ...
\end{source}
The inferred type then takes the simpler form
\[
  \kord{supply\_input} \T \stringty \listty \to (\alpha \hto[\kord{read} \T -\kord{std}] \alpha)
\]

\subsection{References}
\label{sub:references}\enlargethispage{\baselineskip}

Similar to OCaml, mutable references in \Eff are given the effect type $\alpha \kpost{ref}$
with the signature
\[
  \set{
    \kord{lookup} \T \unitty \to \alpha,\;
    \kord{update} \T \alpha \to \unitty
  }
\]
We can define the OCaml accessors by:
\begin{source}
  let (!) r = r#lookup ()
  let (:=) r v = r#update v
\end{source}
with the types
\[
  (\mbox{\texttt{!}}) \T \alpha \kpost{ref}^{\rgn} \to[\kord{lookup} \T \rgn] \alpha \qquad\qquad
  (\mbox{\texttt{:=}}) \T \alpha \kpost{ref}^{\rgn} \to \alpha \to[\kord{update} \T \rgn] \unitty
\]
With reference handlers, we can temporarily alter the value stored in the reference,
make it read-only, log all changes, and more. 
However, handlers are not meant only for overriding but also for defining effects.
In particular, we can use handlers to implement references with a state monad:
\begin{source}
  let state r s0 =
    handler
    | val x -> (fun s -> x)
    | r#lookup () k -> (fun s -> k s s)
    | r#update s' k -> (fun s -> k () s')
    | finally f -> f s0
\end{source}
The function \inline{state} gives a handler that
handles a stateful computation using reference \inline{r} into a pure function that accepts the current state~\inline{s}
and passes it around.
So, we handle \inline{lookup} by a function
that takes the state \inline{s} and passes it to the continuation \inline{k},
which expects the current state as the outcome of \inline{lookup}.
Since this continuation is further handled, it is again a function accepting current state.
Looking up a reference does not change the state, so we pass \inline{s} again,
thus \inline{k} is applied to \inline{s} twice.
In the case for \inline{update},
the expected outcome of the call is the unit value,
while the current state is overwritten by the parameter \inline{s'}.
And the \inline{finally} case says that once we get back a function accepting current state,
we apply it to a given initial state~\inline{s0}.

The inferred type of \inline{state} is
\[
  \kord{state} \T \alpha \kpost{ref}^{\rgn} \to \alpha \to (\beta \hto[\kord{lookup} \T \dotminus\rgn, \kord{update} \T \dotminus\rgn] \beta)
\]
If we want to access the final state, we define \inline{state'} to be exactly the same as \inline{state},
except that its value case is \inline{val x -> (fun s -> (x, s))}.
In this case the inferred type is
\[
  \mbox{\texttt{state'}} \T \alpha \kpost{ref}^{\rgn} \to \alpha \to (\beta \hto[\kord{lookup} \T \dotminus\rgn, \kord{update} \T \dotminus\rgn] \beta \times \alpha)
\]

We can use multiple references without a hitch.
For example, given two references \inline{r1} and \inline{r2}, the computation
\begin{source}
  with (state r1 6) handle
    with (state r2 0) handle
      let x = !r1 in
      r2 := x + 1;
      !r2 * x
\end{source}
returns \inline{42} and its inferred type is the pure $\type{int} \E \emptyset$.

If we replace \inline{state} by \inline{state'},
the handled computation returns \inline{(6, (7, 42))} and the inferred type is $\type{int} \times (\type{int} \times \type{int}) \E \emptyset$.
This shows how easy it is to change the effectful behaviour by just switching the handlers and keeping the imperative code as it is.

\section{Core Eff}
\label{sec:core-eff}

For formal development,
we restrict our attention to \emph{core \Eff}, a minimal fragment of \Eff.
Core \Eff differs from \Eff in the following aspects:
\begin{itemize}
  \item
    Core \Eff is a fine-grain call-by-value calculus~\cite{levy03modelling},
    which means that its terms are split into
    inert \emph{expressions} and possibly effectful \emph{computations}.
    The separation makes the formalization much simpler,
    but makes programming that much harder.
    For this reason, \Eff allows the programmer to freely mix the terms,
    and performs the routine separation automatically.
    For example, a program such as
    \begin{source}
  let transform f m n = (f m 42, n)
    \end{source}
    gets translated into\enlargethispage{\baselineskip}
    \begin{source}
  let transform = fun f -> val (fun m -> val (fun n ->
    let tmp1 =
      let tmp2 = f m in
      tmp2 42
    in
    val (tmp1, n)
  ))
    \end{source}
    where \inline{val} constructs a computation that immediately returns
    a value represented by a given expression.
  \item
    \Eff allows programmers to define their own parametric inductive datatypes and effect types,
    while in core \Eff, we fix the signature of effect types and drop inductive datatypes entirely.
  \item
    \Eff provides a \inline{new} construct
    that allows a programmer to create fresh instances at runtime~\cite{bauer2012programming},
    and can be used to model both exception declarations and reference allocations in ML.
    Since the formalization of this feature is still under investigation, we omit it from our development
    and only briefly discuss it in the Conclusion.
  \item
    Handlers in \Eff allow the additional \inline{finally} case in handlers,
    but as already discussed on page~\pageref{page:finally},
    this is merely a convenience that can be expressed with existing constructs.
\end{itemize}

\subsection{Terms}

We start with a given set of \emph{effects}~$E$,
which are just labels such as $\kord{channel}$, $\kord{exception}$ or $\kord{ref}$
for all possible effects we want to use in our programs.
Next, for each effect~$E$, we have a fixed set $\ops_E$ of operation symbols~$\op$
and a fixed set $\insts_E$ of instances~$\inst$.
We assume that in each operation $\hash{\inst}{\op}$,
both $\inst$ and $\op$ belong to the same effect.

The \emph{expressions}~$e$ and \emph{computations}~$c$ of core \Eff are given by:
\begin{align*}
  \text{expression}~e \bnfis {}
    &x \bnfor
    \tru \bnfor
    \fls \bnfor
    0 \bnfor
    \succ e \bnfor
    \unt \bnfor
    \fun{x} c \bnfor
    \inst \bnfor
    h
    \\  
  \text{handler}~h \bnfis {}
    &(
      \handler
      \val x \mapsto c_v \case
      (\call{e_i}{\op_i}{x}{k} \mapsto c_i)_i
    )
    \\
  \text{computation}~c \bnfis {}
    &\ifthenelse{e}{c_1}{c_2} \bnfor
    \iszero e \bnfor
    \pred e \bnfor
    \absurd e \bnfor
    e_1 \, e_2 \bnfor \\
    &\val e \bnfor
    \call{e_1}{\op}{e_2}{\cont{y}{c}} \bnfor
    \letin{x = c_1} c_2 \bnfor
    \letvalin{x = e} c \bnfor \\
    &\letrecin{f \, x = c_1} c_2 \bnfor
    \withhandle{e}{c}
\end{align*}
Here and everywhere, we write $(-)_{i \in I}$ or just $(-)_i$ to denote a finite repetition of~$-$.
The language constructs are standard except for:
the empty type eliminator $\kord{absurd}$,
the computation $\kord{val}$ that immediately evaluates to a value,
polymorphic let-binding $\kpre{let} \kord{val}$,
and the already discussed \emph{instance constants}, \emph{handlers}, \emph{operation calls} and the \emph{handling construct}.
Note that both instance constants and handlers are first-class values in core \Eff.

The operation calls in core \Eff are slightly different from the ones in \Eff.
The call~$\call{\inst}{\op}{e}{\cont{y}{c}}$ represents an application
of an operation~$\hash{\inst}{\op}$ to a parameter $e$
with a continuation $\cont{y}{c}$ waiting for the result of the call to be bound to~$y$.
Explicit continuations are convenient for operational semantics,
but we do not expect the programmer to use them.
Instead, \Eff uses \emph{generic effects}~\cite{plotkin2003algebraic}, defined as
\[
  \hash{e}{\op} \defeq \fun{x} \call{e}{\op}{x}{\cont{y}{\val y}}
\]
which take a parameter and perform the operation call with the trivial continuation.
Then, the programmer can write the more familiar $\letin{y = \hash{\inst}{\op} \, e} c$
instead of $\call{\inst}{\op}{e}{\cont{y}{c}}$,
and we shall see that the two exhibit equivalent behaviour.  

The rough idea is that each non-divergent computation either evaluates to a value or calls an operation.
We use a handler~$h = \handler \val x \mapsto c_v \case (\call{\inst_i}{\op_i}{x}{k} \mapsto c_i)_i$ on a computation~$c$ as follows:
\begin{itemize}
\item
  If $c$ evaluates to a value $\val e$, we use the value case and evaluate $c_v[e / x]$.
\item
  If $c$ performs an operation call~$\call{\inst_j}{\op_j}{e}{\cont{y}{c'}}$
  and the handler contains a matching case $\call{\inst_j}{\op_j}{x}{k} \mapsto c_j$,
  we evaluate $c_j[e / x, (\fun{y} \withhandle{h}{c'}) / k]$.
  We wrap the handler $h$ around the continuation so that it may continue handling future operation calls,
  though any operations called by $c_j$ escape its scope.
\item
  If the handler has no matching operation cases for the called operation,
  then just like in exception handlers,
  we propagate the call outwards for other handlers to catch,
  though we still wrap $h$ around the continuation as it may handle some other operations.
\end{itemize}

\noindent Let-binding $\letin{x = c_1} c_2$ works as follows:
if $c_1$ evaluates to $\val e$, we continue with $c_2[e / x]$,
but if $c_1$ calls an operation,
we propagate the call outwards just like when a handler has no matching cases.
In fact, let-binding $\letin{x = c_1} c_2$ works exactly as
$\withhandle{(\handler \val x \mapsto c_2)}{c_1}$
Though this makes let binding redundant,
we keep it in the language for convenient notation and
to serve as a stepping stone to the less familiar handling construct.

To make the above intuition more precise and to motivate the effect system,
we now give a small-step operational semantics,
determined by a relation $c \step c'$ defined in Figure~\ref{fig:small-step},
stating that a computation~$c$ takes a single step to $c'$.
Note that the relation is given for computations only and that expressions are inert.

\begin{figure}[h]
\hrulefill
  \small
  \begin{mathpar}
  \inferrule{
  }{
    \ifthenelse{\tru}{c_1}{c_2} \step c_1
  }

  \inferrule{
  }{
    \ifthenelse{\fls}{c_1}{c_2} \step c_2
  }

  \inferrule{
  }{
    \iszero 0 \step \val \tru
  }

  \inferrule{
  }{
    \iszero (\succ e) \step \val \fls
  }

  \inferrule{
  }{
    \pred 0 \step \val 0
  }

  \inferrule{
  }{
    \pred (\succ e) \step \val e
  }

  \inferrule{
  }{
    (\fun{x} c) \, e \step c[e / x]
  }

  \inferrule{
    c_1 \step c_1'
  }{
    \letin{x = c_1} c_2 \step \letin{x = c_1'} c_2
  }

  \inferrule{
  }{
    \letin{x = \val e} c \step c[e / x]
  }

  \inferrule{
  }{
    \letin{x = \call{\inst}{\op}{e}{\cont{y}{c_1}}} c_2
      \step \call{\inst}{\op}{e}{\cont{y}{\letin{x = c_1} c_2}}      
  }

  \inferrule{
  }{
    \letvalin{x = e} c \step c[e / x]
  }

  \inferrule{
  }{
    \letrecin{f \, x = c_1} c_2
      \step c_2[(\fun{x} \letrecin{f \, x = c_1} c_1) / f]
  }

  \inferrule{
    c \step c'
  }{
    \withhandle{e}{c} \step \withhandle{e}{c'}
  }

  \inferrule{
  }{
    \withhandle{h}{(\val e)} \step c_v[e / x]
  }

  \inferrule{
  }{
    \withhandle{h}{(\call{\inst_j}{\op_j}{e}{\cont{y}{c}})}
      \step c_j[e / x, (\fun{y} \withhandle{h}{c}) / k]
  }

  \inferrule{
    \hash{\inst}{\op} \not\in \set{\hash{\inst_i}{\op_i}}_i
  }{
    \withhandle{h}{(\call{\inst}{\op}{e}{\cont{y}{c}})}
      \step \call{\inst}{\op}{e}{\cont{y}{\withhandle{h}{c}}}
  }
\end{mathpar}
\caption{
  The inductive definition of the relation $c \step c'$.
  In the last three rules, we set $h = \handler
    \val x \mapsto c_v \case
    (\call{\inst_i}{\op_i}{x}{k} \mapsto c_i)_i$.}
\label{fig:small-step}
\hrulefill
\end{figure}

\begin{exa}
\label{exa:reduction}
Take a reference handler
\begin{align*}
  h \defeq {} &\handler \\
  &\case \val x \mapsto \hash{r}{\kord{update}} \, x \\
  &\case \call{r}{\kord{lookup}}{x}{k} \mapsto k \, (\succ 0) \\
  &\case \call{r}{\kord{update}}{x}{k} \mapsto k \, \unt
\end{align*}
which temporarily treats a reference $r \in \insts_{\kord{ref}}$ as if it always contains $1$,
and afterwards updates it with the final result of the handled computation.
This update is not handled by~$h$ because it escapes its scope.
If we apply $h$ on the computation
\begin{align*}
  c \defeq {} &\letin{x_1 = \\
  &\quad \letin{x_2 = \hash{r}{\kord{lookup}} \, \unt} \\
  &\quad \hash{r}{\kord{update}} \, x_2 \\
  &} \val 0
\end{align*}
the outcome of the first lookup is $1$, which is then bound to $x_2$,
while the handler continues handling the continuation.
Then, the update is ignored and finally,
the handler applies the value case on $0$ and terminates with a call of $\hash{r}{\kord{update}}$.
Note that \Eff provides \emph{resources}~\cite{bauer2012programming},
which at this point trigger real-world effects and resume the continuation.
The exact reduction sequence is given in Figure~\ref{fig:reduction}.
\end{exa}

\begin{figure}[h]
\hrulefill
  \small
\newcommand{\hilite}{\underline}
\begin{align*}
  &\withhandle{h}{\\
  &\quad \letin{x_1 = (\letin{x_2 = \hilite{\hash{r}{\kord{lookup}} \, \unt}} \hash{r}{\kord{update}} \, x_2)} \val 0} \step \\
  &\withhandle{h}{\\
  &\quad \letin{x_1 = (\letin{\hilite{x_2} = \hilite{\call{r}{\kord{lookup}}{\unt}{\cont{y_1}{\val y_1}}}} \hash{r}{\kord{update}} \, x_2)} \val 0} \step \\
  &\withhandle{h}{\\
  &\quad \letin{\hilite{x_1} = \hilite{\call{r}{\kord{lookup}}{\unt}{\cont{y_1}{\letin{x_2 = \val y_1} \hash{r}{\kord{update}} \, x_2}}}} \val 0} \step \\
  &\withhandle{\hilite{h}}{\\
  &\quad \hilite{\call{r}{\kord{lookup}}{\unt}{\cont{y_1}{\letin{x_1 = (\letin{x_2 = \val y_1} \hash{r}{\kord{update}} \, x_2)} \val 0}}}} \step \\
  &\hilite{(\fun{y_1} \withhandle{h}{(\letin{x_1 = (\letin{x_2 = \val y_1} \hash{r}{\kord{update}} \, x_2)} \val 0)})} \,\, \hilite{1\vphantom{()}} \step \\
  &\withhandle{h}{(\letin{x_1 = (\letin{\hilite{x_2} = \hilite{\val 1}} \hash{r}{\kord{update}} \, x_2)} \val 0)} \step \\
  &\withhandle{h}{(\letin{x_1 = \hilite{\hash{r}{\kord{update}} \, 1}} \val 0)} \step \\
  &\withhandle{h}{(\letin{\hilite{x_1} = \hilite{\call{r}{\kord{update}}{1}{\cont{y_2}{\val y_2}}}} \val 0)} \step \\
  &\withhandle{\hilite{h}}{(\hilite{\call{r}{\kord{update}}{1}{\cont{y_2}{\letin{x_1 = \val y_2} \val 0}}})} \step \\
  &\hilite{(\fun{y_2} \withhandle{h}{(\letin{x_1 = \val y_2} \val 0)})} \,\, \hilite{\unt\vphantom{()}} \step \\
  &\withhandle{h}{(\letin{\hilite{x_1} = \hilite{\val \unt}} \val 0)} \step \\
  &\withhandle{\hilite{h}}{\hilite{(\val 0)}} \step
  \hilite{\hash{r}{\kord{update}} \, 0} \step \call{r}{\kord{update}}{0}{\cont{y_3}{\val y_3}}
\end{align*}
\caption{
  The evaluation of $\withhandle{h}{c}$, as described in Example~\ref{exa:reduction}.
  We underline the active parts of each step and shorten $\succ 0$ to $1$.
  We can see how the operation call to $\hash{r}{\kord{lookup}}$ in the first line propagates
  outwards to the matching handler while its continuation builds up.
  Once the call reaches a handler, it is replaced with the handling term
  in which $k$ is replaced by the further handled continuation.}
\label{fig:reduction}
\hrulefill
\end{figure}

\subsection{Types}
\label{sub:types}

The types, which are also split into pure and potentially effectful (here called \emph{dirty}) ones, are given by
\begin{align*}
  \text{type}\ A, B &\bnfis
    \boolty \bnfor
    \natty \bnfor
    \unitty \bnfor
    \emptyty \bnfor
    A \to \C \bnfor
    E^\Rgn \bnfor
    \C \hto \D
    \\
  \text{dirty type}~\C, \D &\bnfis
    A \E \Drt
\end{align*}
We have the usual ground types and the function type~$A \to \C$
of functions that take expressions of type~$A$ and perform computations of type~$\C$.
Next, we have the \emph{effect type}~$E^\Rgn$ of instances of effect~$E$ from a \emph{region}~$\Rgn$,
which is just a \emph{non-empty} finite set of instances $\set{\inst_1, \dots, \inst_n} \subseteq \insts_E$.
Finally, we have the \emph{handler type}~$\C \hto \D$ of handlers that take a computation of type~$\C$ and transform it into a computation of type~$\D$.
We call $\C$ the \emph{incoming} and $\D$ the \emph{outgoing} type.
Finally, dirty type~$A \E \Drt$ contains computations that
return values of type $A$ and may cause effects described by a \emph{dirt}~$\Drt$,
which is a set of operations
\[
  \set{\hash{\inst_1}{\op_1}, \dots, \hash{\inst_n}{\op_n}}
\]
To lighten the syntax, we write $A_1 \to A_2 \E \Drt$ as $A_1 \to[\Drt] A_2$
where we also omit the outer braces around $\Drt$.

Note that for simplicity, the types of core \Eff are monomorphic.
However, we are going to shift to polymorphic types with type, region and dirt parameters
when we start with type inference in Section~\ref{sec:inferring}.

\subsection{Subtyping}
\label{sub:subtyping}

As in most effect systems,
we need to take care of the \emph{poisoning problem}~\cite{wansbrough1999once}.
For example, what type should we give to~$\mathit{ignore}$ in
\begin{align*}
  &\letin{\mathit{ignore} = \val (\fun{msg} \val \unt)} \\
  &\letin{f = \ifthenelse{b}{(\val \mathit{ignore})}{(\val \hash{\kord{std}}{\kord{write}}})} \\
  &\val \mathit{ignore}
\end{align*}
for some suitable boolean $b$?
If we give it the type~$\type{string} \to \unitty$
(for this example, we allow ourselves the type $\type{string}$ of strings),
we cannot type the conditional statement as the two branches cannot have the same type,
but if we give it the type~$\type{string} \to[\hash{\kord{std}}{\kord{write}}] \unitty$,
we lose information that the final result is a pure function.

The simplest antidote for the poisoning problem is to allow subtyping,
so that we may give $\kord{ignore}$ the type with an empty dirt,
and use subsumption to suitably enlarge this dirt in the conditional statement.

Subtyping also solves a similar problem with regions of handled instances.
Consider the computation
\begin{align*}
  &\letin{u = \val \inst} \\
  &\letin{v = \ifthenelse{b}{\val u}{\val \inst'}} \\
  &\val (\handler \val x \mapsto \cdots \case \call{u}{\op}{x}{k} \mapsto \cdots)
\end{align*}
Without subtyping we are forced to give both $u$ and $v$ the type $E^{\set{\inst, \inst'}}$.
Therefore, as discussed in Section~\ref{ssub:handled-regions},
the type of $u$ does not tell us whether $h$ handles
$\hash{\inst}{\op}$ or $\hash{\inst'}{\op}$, and so we must assume that both may remain unhandled.
With subtyping we may give $u$ the type $E^{\set{\inst}}$, which makes it clear that $h$ handles $\hash{\inst}{\op}$.

For our purposes, it is enough to use \emph{structural} subtyping~\cite{fuh1990type},
where we relate only types of the same shape.
The subtyping relations between types and between dirty types are defined in Figure~\ref{fig:subtyping}.
\begin{figure}
\hrulefill
  \small
  \begin{mathpar}
  \inferrule[Sub-$\boolty$]{
  }{
    \boolty \le \boolty
  }

  \inferrule[Sub-$\natty$]{
  }{
    \natty \le \natty
  }

  \inferrule[Sub-$\unitty$]{
  }{
    \unitty \le \unitty
  }

  \inferrule[Sub-$\emptyty$]{
  }{
    \emptyty \le \emptyty
  }

  \inferrule[Sub-$\to$]{
    A' \le A \\
    \C \le \C'
  }{
    A \to \C \le A' \to \C'
  }

  \inferrule[Sub-$E$]{
    \Rgn \subseteq \Rgn'
  }{
    E^\Rgn \le E^{\Rgn'}
  }

  \inferrule[Sub-$\hto$]{
    \C' \le \C \\
    \D \le \D'
  }{
    \C \hto \D \le \C' \hto \D'
  }

  \inferrule[Sub-$\E$]{
    A \le A' \\
    \Drt \subseteq \Drt'
  }{
    A \E \Drt \le A' \E \Drt'
  }
\end{mathpar}
\caption{
  The inductive definition of the subtyping relations $A \le A'$ and $\C \le \C'$.}
\label{fig:subtyping}
\hrulefill
\end{figure}
Sometimes, we shall be interested in types that have the same shape.
So, we define $\approx$ as the equivalence relation on types, generated by $\le$.
Equivalence classes of $\approx$ are called \emph{skeletons}~\cite{simonet2003type}.

\subsection{Effect system}
\label{sub:effect-system}

Our effect system is built on two typing judgements, defined in Figure~\ref{fig:typing}.
The judgement $\ctx \ent[\sig] e \T A$ states that in context~$\ctx$ and \emph{signature}~$\sig$, an expression~$e$ has a type~$A$.
The judgement $\ctx \ent[\sig] c \T \C$ states a similar thing for a computation~$c$ and a dirty type~$\C$.
In both cases, the context $\ctx$ is a unique assignment of (pure) types to variables,
while the signature $\sig$ consists of \emph{effect signatures}~$\sig(E)$ for each effect $E$.
These are of the form
\[
  \set{\op_1 \T A^{\op_1} \to B^{\op_1}, \dots, \op_n \T A^{\op_n} \to B^{\op_n}}
\]
and assign a \emph{parameter type}~$A^\op$ and a \emph{result type}~$B^\op$ to each listed operation~$\op$.
For example, the effect signatures for references is:
\[
  \Sigma(\kord{ref}) = \set{\kord{lookup} \T \unitty \to \natty, \kord{update} \T \natty \to \unitty}
\]
For technical reasons, we assume that both the parameter and the result type for each operation do not contain any regions or dirt,
which limits them to the basic ground types such as $\natty$ or $\unitty$.
We further discuss this restriction in Remark~\ref{rem:glitch}.

The purpose of the presented effect system is to offer guarantees on the behaviour of programs,
not (yet) to lead to an efficient inference algorithm.
One sign of that is rules like \rulename{Val}, \rulename{Inst} or \rulename{Pred},
where we assign types that are safe, but much coarser than needed.
A more obvious sign is the rule \rulename{LetVal},
where we employ a very naive form of let-polymorphism that performs an explicit substitution.
This is, of course, extremely inefficient, but lets us postpone the use of parameters to
the inference rules, which use the more efficient variant with universally quantified types.

\begin{figure}
\hrulefill
  \small
  \begin{mathpar}
  \inferrule[Var]{
    (x \T A) \in \ctx
  }{
    \ctx \ent x \T A
  }

  \inferrule[True]{
  }{
    \ctx \ent \tru \T \boolty
  }

  \inferrule[False]{
  }{
    \ctx \ent \fls \T \boolty
  }

  \inferrule[Zero]{
  }{
    \ctx \ent 0 \T \natty
  }

  \inferrule[Succ]{
    \ctx \ent e \T \natty
  }{
    \ctx \ent \succ e \T \natty
  }

  \inferrule[Unit]{
  }{
    \ctx \ent \unt \T \unitty
  }

  \inferrule[Fun]{
    \ctx, x \T A \ent c \T \C
  }{
    \ctx \ent \fun{x} c \T A \to \C
  }

  \inferrule[Inst]{
    \inst \in \Rgn \subseteq \insts_E
  }{
    \ctx \ent \inst \T E^\Rgn
  }

  \text{\rulename{Hand} --- in the text}

  \inferrule[SubExpr]{
    \ctx \ent e \T A \\
    A \le A'
  }{
    \ctx \ent e \T A'
  }

  \medskip

  \inferrule[IfThenElse]{
    \ctx \ent e \T \boolty \\
    \ctx \ent c_1 \T \C \\
    \ctx \ent c_2 \T \C
  }{
    \ctx \ent \ifthenelse{e}{c_1}{c_2} \T \C
  }

  \inferrule[IsZero]{
    \ctx \ent e \T \natty
  }{
    \ctx \ent \iszero e \T \boolty \E \Drt
  }

  \inferrule[Pred]{
    \ctx \ent e \T \natty
  }{
    \ctx \ent \pred e \T \natty \E \Drt
  }

  \inferrule[Absurd]{
    \ctx \ent e \T \emptyty
  }{
    \ctx \ent \absurd e \T \C
  }
  
  \inferrule[App]{
    \ctx \ent e_1 \T A \to \C \\
    \ctx \ent e_2 \T A
  }{
    \ctx \ent e_1 \, e_2 \T \C
  }

  \inferrule[Val]{
    \ctx \ent e \T A
  }{
    \ctx \ent \val e \T A \E \Drt
  }

  \inferrule[Op]{
    \ctx \ent e_1 \T E^\Rgn \\
    \op \T A^\op \to B^\op \in \sig(E) \\\\
    \ctx \ent e_2 \T A^\op \\
    \ctx, y \T B^\op \ent c \T A \E \Drt \\
    \fra{\inst \in \Rgn} \hash{\inst}{\op} \in \Drt
  }{
    \ctx \ent \call{e_1}{\op}{e_2}{\cont{y}{c}} \T A \E \Drt
  }

  \inferrule[Let]{
    \ctx \ent c_1 \T A \E \Drt \\
    \ctx, x \T A \ent c_2 \T B \E \Drt
  }{
    \ctx \ent \letin{x = c_1} c_2 \T B \E \Drt
  }

  \inferrule[LetVal]{
    \ctx \ent e \T A \\
    \ctx \ent c[e / x] \T \C
  }{
    \ctx \ent \letvalin{x = e} c \T \C
  }

 \inferrule[LetRec]{
    \ctx, f \T A \to \C, x \T A \ent c_1 \T \C \\
    \ctx, f \T A \to \C \ent c_2 \T \D
  }{
    \ctx \ent \letrecin{f \, x = c_1} c_2 \T \D
  }

  \inferrule[With]{
    \ctx \ent e \T \C \hto \D \\
    \ctx \ent c \T \C
  }{
    \ctx \ent \withhandle{e}{c} \T \D
  }

  \inferrule[SubComp]{
    \ctx \ent c \T \C \\
    \C \le \C'
  }{
    \ctx \ent c \T \C'
  }
\end{mathpar}
\caption{
  The inductive definition of the typing judgements $\ctx \ent[\sig] e \T A$ and $\ctx \ent[\sig] c \T \C$.
  As the signature~$\sig$ does not change, we omit its display from all the judgements.
  The rule for handlers is given in the main text.}
\label{fig:typing}
\hrulefill
\end{figure}

All the typing rules are standard except for:
\begin{description}
\item[\rulename{Inst}]
  in which we check that $\inst$ is contained in the region $\Rgn$ that belongs to an effect $E$.
\item[\rulename{Op}]
  in which we first check that $e_1$ and $\op$ belong to the same effect.
  Then, we need to check that the dirt $\Drt$ covers not just all possible operations that the operation call may cause
  (recall that $\Rgn$ may contain more than one instance), but also any operations in the continuation~$c$.
  We may assume that $c$ has the same dirt, as we can use \rulename{SubComp} otherwise.
  We use the same reasoning in rules \rulename{IfThenElse} and \rulename{Let}.
\item[\rulename{With}]
  where the handling construct is typed like an application,
  except that it is applied to a computation rather than an expression.
\item[\rulename{Hand}]
  which is a bit more daunting, so we write it out separately:
  \[
    \inferrule[Hand]{
      \ctx, x \T A \ent c_v \T B \E \Drt' \\
      (\prms_i)_i \\
      \prms_\Drt
    }{
      \ctx \ent (
        \handler
        \val x \mapsto c_v \case
        (\call{e_i}{\op_i}{x}{k} \mapsto c_i)_i
      ) \T A \E \Drt \hto B \E \Drt'
    }
  \]
  For a handler to be of type $A \E \Drt \hto B \E \Drt'$,
  we first check that it takes values of type $A$ to computations of type $B \E \Drt'$.
  Then, for each operation case $\call{e_i}{\op_i}{x}{k} \mapsto c_i$, we check the premises~$\prms_i$,
  comprising:
  \[
    \ctx \ent e_i \T E_i^{\Rgn_i} \qquad
    \op_i \T A^{\op_i} \to B^{\op_i} \in \sig(E_i) \qquad
    \ctx, x \T A^{\op_i}, k \T B^{\op_i} \to B \E \Drt' \ent c_i \T B \E \Drt'
  \]
  Like in \rulename{Op}, we check that $e_i$ and $\op_i$ belong to the same effect $E_i$.
  Then, the handling computation~$c_i$ needs to have the same type $B \E \Drt'$,
  assuming that parameter $x$ is of type $A^{\op_i}$,
  and the continuation~$k$ of type~$B^{\op_i} \to B \E \Drt'$.
  Observe that since the continuation is further handled, it already has the outgoing type.

  Finally, in $\prms_\Drt$ we check that any operation in the incoming dirt~$\Drt$
  that is not guaranteed to be caught by the handler
  must appear in the outgoing dirt~$\Drt'$ as well.
  As discussed in~Section~\ref{ssub:handled-regions},
  an operation $\hash{\inst}{\op}$ will be (if not sooner)
  surely caught by the case for $\hash{\inst_i}{\op_i}$
  whenever $\op = \op_i$ and the region~$\Rgn_i$ is the singleton~$\set{\inst_i}$.
  Thus, we define $\prms_\Drt$ to be
  \[
    \fra{\hash{\inst}{\op} \in \Drt} (\inst \in \varuniq{\op = \op_i}{\Rgn_i}) \lor (\hash{\inst}{\op} \in \Drt')
  \]
  where the \emph{singleton union} $\cupdot$ behaves like a union, except that it considers only singletons.
  For example, $\set{\inst_1, \inst_2} \cupdot \set{\inst_2} \cupdot \set{\inst_3, \inst_4} \cupdot \set{\inst_5} = \set{\inst_2, \inst_5}$.
  More precisely, we define:
  \[
    \uniq{i}{\Rgn_i} = \set{\inst \mid \set{\inst} \in \set{\Rgn_i}_i}
  \]
\end{description}\medskip

\noindent The given effect system is then safe with respect to the operational semantics:
a computation $\ent c \T A \E \Drt$ can only call operations from $\Drt$.
In particular, if $\Drt$ is empty, then $c$ is guaranteed to be pure, though it may diverge.

\begin{thm}[Safety]
\label{thm:safety}
If for a computation~$c$, the typing judgement $\ent c \T A \E \Drt$ holds, then either:
\begin{itemize}
\item
  $c$ is of the form $\val e$ for some expression $\ent e \T A$, or
\item
  $c$ is of the form $\call{\inst}{\op}{e}{\cont{y}{c'}}$ for some $\hash{\inst}{\op} \in \Drt$, or
\item
  there exists a computation $c'$ such that $c \step c'$ and $\ent c' \T A \E \Drt$.
\end{itemize}
\end{thm}

We do not give a proof of Theorem~\ref{thm:safety} in this paper.
Instead, a full formalization of core \Eff in Twelf is available at~\url{https://github.com/matijapretnar/twelf-eff/}

\begin{exa}
To illustrate the type system, let us revisit Example~\ref{exa:reduction}.
First, assume that the reference $r \in \insts_{\kord{ref}}$ is given the type $\kord{ref}^{\set{r}}$.
Then, the stateful computation~$c$ has the dirty type
$\natty \E \set{\hash{r}{\kord{lookup}}, \hash{r}{\kord{update}}}$,
while the handler~$h$ has the type
\[
  (\natty \E \set{\hash{r}{\kord{lookup}}, \hash{r}{\kord{update}}}) \hto (\unitty \E \set{\hash{r}{\kord{update}}})
\]
as it handles both $\kord{lookup}$ and $\kord{update}$, but then triggers $\kord{update}$ in the value case.
This $\kord{update}$ also changes the type of computation from $\natty$ to $\unitty$.

If the reference $r$ has a less precise type $\kord{ref}^{\set{r, r'}}$,
the dirt of $c$ is 
\[
  \Drt \defeq \set{\hash{r}{\kord{lookup}}, \hash{r'}{\kord{lookup}}, \hash{r}{\kord{update}}, \hash{r'}{\kord{update}}}
\]
while the best type we can give to $h$ is $(\natty \E \Drt) \hto (\unitty \E \Drt)$.
Since the region of $r$ is not a singleton,
we unfortunately cannot give any guarantees on the handled operations.
\end{exa}

\section{Inferring constraints}
\label{sec:inferring}

Turning to our inference algorithm, we first describe
a collection of syntax-directed inference rules
that are readily transcribed into a recursive function that
infers a type and a set of constraints from a given term.

\subsection{Parametric types}
\label{sub:parametric-types}

As indicated in Section~\ref{sub:types}, we switch to a language that is more suited for inference:
\begin{align*}
  \text{type}~A, B &\bnfis
    \alpha \bnfor
    \boolty \bnfor
    \natty \bnfor
    \unitty \bnfor
    \emptyty \bnfor
    A \to \C \bnfor
    E^{\rgn*} \bnfor
    \C \hto \D
    \\
  \text{dirty type}~\C, \D &\bnfis
    A \E \Drt
    \\
  \text{dirt}~\Drt &\bnfis
    \set{\op_1 \T \rgn_1, \dots, \op_n \T \rgn_n \mid \drt^\ops}
\end{align*}
From now on, we refer to types and dirt from Section~\ref{sub:types},
which contain no parameters, as \emph{closed} ones.

To enable polymorphism, types are extended with the \emph{type parameters}~$\alpha$.
Then, regions are not just extended,
but completely replaced with \emph{region parameters}~$\rgn$,
as this greatly simplifies the inference.
We are going to capture the information about instances using constraints instead.
Recall that regions~$\Rgn$ describing the possible instances in an effect type $E^\Rgn$ are always inhabited.
This information will prove useful in Section~\ref{sec:simplifying} as it allows further simplification.
For this reason, we designate a special subset of region parameters, called \emph{inhabited} and marked by~$\rgn*$.

Finally, we adopt a row-like~\cite{remy1993type} representation of dirt as described in Section~\ref{ssub:raising-exceptions}.
The first part is a set of operation symbols together with a region parameter
that captures the (possibly empty) region of all the instances on which these symbols are used.
This is similar to before, except that operations are grouped by their operation symbols.
The reason for this grouping is that we are always able to precisely determine the operation symbols,
but not the instances of called or handled operations.

The second part consists of a single \emph{dirt parameter}~$\drt$,
intended to capture the rest of operations.
If the first part is empty, we write the dirt simply as $\drt$.
To keep track of the operation symbols captured by $\drt$
and ensure that it does not capture any symbols from the first part,
we sometimes write the parameter as $\drt^\ops$
to emphasise the set~$\ops$ of symbols not captured by $\drt$
(though $\ops$ can always be reconstructed by looking
at the first part of any dirt in which $\drt$ appears
because our algorithm ensures that all such dirts consistently list the same operation symbols).

Any additional information is captured with \emph{constraints}.
For example, a conditional can call any operation that one of its branches does.
If these two branches cause dirt captured by $\drt_1$ and $\drt_2$,
we can represent the dirt of the whole conditional with a fresh parameter $\drt$
together with constraints $\drt_1 \le \drt$ and $\drt_2 \le \drt$.
Constraints can be of one of the following five kinds:
\begin{itemize}
\item
  $A \le A'$ states that the type $A$ needs to be smaller than $A'$,
\item
  $\C \le \C'$ states the same for dirty types,
\item
  $\rgn \le \rgn' \cup \uniq{i}{\rgn*_i}$ is a generalisation of the inequality $\rgn \le \rgn'$ due to handlers.
  It states that all instances from $\rgn$ are either in $\rgn'$ or in some $\rgn*_i$ that is a singleton,
\item
  $\inst \in \rgn \cup \uniq{i}{\rgn*_i}$ similarly states that $\inst$ is either in $\rgn$ or in some $\rgn*_i$ that is a singleton,
\item
  $\Drt \le \Drt'$ states that the dirt $\Drt$ is smaller than $\Drt'$.
\end{itemize}
In the right-hand side $\rgn \cup \uniq{i}{\rgn*_i}$ of two region constraints,
we refer to $\rgn$ as the \emph{covering},
and to $\rgn*_i$ as the \emph{handled} region parameters.

Unlike the subtyping relation, constraints do not have any inherent reasoning principles,
but are just a way of writing down the relationship between parameters.
Instead, we give constraints a meaning by specifying their solutions.

\begin{defi}
\label{defi:closed-substitution}
A \emph{closed substitution}~$\sol$ is a partial mapping that maps:
  each type parameter~$\alpha$ to a closed type~$\sol(\alpha)$,
  each region parameter~$\rgn$ to a closed region~$\sol(\rgn)$,
  each inhabited region parameter~$\rgn*$ to a non-empty closed region~$\sol(\rgn*)$,
  and each dirt parameter~$\drt^\ops$ to a closed dirt~$\sol(\drt^\ops)$
  that contains no operations with operation symbols in~$\ops$.
\end{defi}

We write substitutions by specifying a set of mappings of parameters, for example
\[
  \set{\alpha_1 \mapsto \natty, \alpha_2 \mapsto \unitty, \rgn \mapsto \emptyset, \drt \mapsto \set{\hash{\inst}{\op}}}
\]
We can extend a closed substitution to other constructs by:
\begin{gather*}
\begin{aligned}
  \sol(\boolty) &= \boolty \qquad&
  \sol(A \to \C) &= \sol(A) \to \sol(\C) \\
  \sol(\natty) &= \natty &
  \sol(E^{\rgn*}) &= E^{\sol(\rgn*)} \\
  \sol(\unitty) &= \unitty &
  \sol(\C \hto \D) &= \sol(\C) \hto \sol(\D) \\
  \sol(\emptyty) &= \emptyty &
  \sol(A \E \Drt) &= \sol(A) \E \sol(\Drt) \\
\end{aligned} \\
\sol(\set{\op_1 \T \rgn_1, \dots, \op_n \T \rgn_n \mid \drt})
  = \Big( \bigcup_{i = 1}^n \set{\hash{\inst}{\op_i} \mid \inst \in \sol(\rgn_i)} \Big)
     \cup \sol(\drt)
\end{gather*}

\begin{defi}
A closed substitution~$\sol$ is a \emph{solution} of a set of constraints~$\cstr$,
which we write as $\sol \models \cstr$,
if it satisfies all the constraints in $\cstr$.
This is defined by:
\begin{align*}
  \sol \models A \le A' &\iff
    \sol(A) \le \sol(A') \\
  \sol \models \C \le \C' &\iff
    \sol(\C) \le \sol(\C') \\
  \sol \models \rgn \le \rgn' \cup \uniq{i}{\rgn*_i} &\iff
    \sol(\rgn) \subseteq \sol(\rgn') \cup \uniq{i}{\sol(\rgn*_i)}  \\
  \sol \models \inst \in \rgn \cup \uniq{i}{\rgn*_i} &\iff
    \inst \in \sol(\rgn) \cup \uniq{i}{\sol(\rgn*_i)} \\
  \sol \models \Drt \le \Drt' &\iff
    \sol(\Drt) \subseteq \sol(\Drt')
\end{align*}
\end{defi}

A parametric type $A$ together with a set of constraints $\cstr$ between its parameters
then describes a family of closed types, obtained by taking all instances $\sol(A)$ of $A$
for all solutions $\sol \models \cstr$,
and all their supertypes.
More precisely, we define
\[
  \types{\cstr} \defeq \set{A' \mid \sol(A) \le A', \sol \models \cstr} \qquad\qquad
  \types[\C]{\cstr} \defeq \set{\C' \mid \sol(\C) \le \C', \sol \models \cstr}
\]
Our aim now is to take an expression $e$
and compute a type $A$ and a set of constraints $\cstr$
such that the set of all possible types $A'$ we can assign to $e$ is captured exactly by $\types{\cstr}$.

\subsection{Inference rules}
\label{sub:inference-rules}

We infer types and constraints using syntax-directed inference rules of the form
$\ctx; \pctx \ent[F] e \T A \while \cstr$ for expressions and
$\ctx; \pctx \ent[F] c \T \C \while \cstr$ for computations,
defined in Figure~\ref{fig:constraints}.
Here, $\cstr$ is a set of constraints,
$F$ is a set of all fresh parameters introduced in the derivation,
and $\pctx$ is the \emph{polymorphic context}, which is collection
of unique assignments $x_i \T \fra{F_i} A_i \while \cstr_i$ of
\emph{type schemes} to variables (assumed to be different from the ones in $\ctx$).
As in typing judgements, we assume (though never write) a fixed signature~$\sig$.

The type schemes are similar to polymorphic types of ML, which are types,
universally quantified over a given set of type parameters.
In our case, we may also quantify over region and dirt parameters,
but we need to keep information about the constraints these parameters need to satisfy.
Even though the ordinary context $\ctx$ can be seen as a particular instance of $\pctx$,
we keep the two separate in order to relate the inference judgements to typing judgements,
as the latter employ only $\ctx$.

Though $\cstr$ and $F$ are sets, we sometimes write them as sequences to save space.
For example, we write $F_1, F_2, \alpha, \drt$ instead of $F_1 \cup F_2 \cup \set{\alpha, \drt}$.
We also assume that all parameters listed in $F$ are distinct
and this implies the usual freshness conditions~\cite[p.~321]{pierce2002types}.
For example, the above sequence implies that sets $F_1$ and $F_2$
are disjoint and do not contain $\alpha$ or $\drt$.
In particular, in the rule \rulename{Cstr-PolyVar},
we implicitly rename any bound parameters $F$ so that a fresh copy is obtained at each use.

\begin{figure}
\hrulefill
  \small
  \begin{mathpar}
  \inferrule[Cstr-Var]{
    (x \T A) \in \ctx
  }{
    \ctx; \pctx \ent[\emptyset] x \T A \while \emptyset
  }

  \inferrule[Cstr-PolyVar]{
    \big(x \T \fra{F} A \while \cstr\big) \in \pctx
  }{
    \ctx; \pctx \ent[F] x \T A \while \cstr
  }

  \inferrule[Cstr-True]{
  }{
    \ctx; \pctx \ent[\emptyset] \tru \T \boolty \while \emptyset
  }

  \inferrule[Cstr-False]{
  }{
    \ctx; \pctx \ent[\emptyset] \fls \T \boolty \while \emptyset
  }

  \inferrule[Cstr-Zero]{
  }{
    \ctx; \pctx \ent[\emptyset] 0 \T \natty \while \emptyset
  }

  \inferrule[Cstr-Succ]{
    \ctx; \pctx \ent[F] e \T A \while \cstr
  }{
    \ctx; \pctx \ent[F] \succ e \T \natty \while \cstr, A \le \natty
  }

  \inferrule[Cstr-Unit]{
  }{
    \ctx; \pctx \ent[\emptyset] \unt \T \unitty \while \emptyset
  }

  \inferrule[Cstr-Fun]{
    \ctx, x \T \alpha; \pctx \ent[F] c \T \C \while \cstr
  }{
    \ctx; \pctx \ent[F, \alpha] \fun{x} c \T \alpha \to \C \while \cstr
  }

  \inferrule[Cstr-Inst]{
  }{
    \ctx; \pctx \ent[\rgn*] \inst \T E^{\rgn*} \while \inst \in \rgn*
  }

  \text{\rulename{Cstr-Hand} --- in the text}

  \inferrule[Cstr-IfThenElse]{
    \ctx; \pctx \ent[F] e \T A \while \cstr \\
    \ctx; \pctx \ent[F_1] c_1 \T \C_1 \while \cstr_1 \\
    \ctx; \pctx \ent[F_2] c_2 \T \C_2 \while \cstr_2
  }{
    \ctx; \pctx \ent[F, F_1, F_2, \alpha, \drt] \ifthenelse{e}{c_1}{c_2} \T \alpha \E \drt \while
      \cstr, \cstr_1, \cstr_2, A \le \boolty, \C_1 \le (\alpha \E \drt), \C_2 \le (\alpha \E \drt)
  }

  \inferrule[Cstr-IsZero]{
    \ctx; \pctx \ent[F] e \T A \while \cstr
  }{
    \ctx; \pctx \ent[F, \drt] \iszero e \T \boolty \E \drt \while \cstr, A \le \natty
  }

  \inferrule[Cstr-Pred]{
    \ctx; \pctx \ent[F] e \T A \while \cstr
  }{
    \ctx; \pctx \ent[F, \drt] \pred e \T \natty \E \drt \while \cstr, A \le \natty
  }

  \inferrule[Cstr-Absurd]{
    \ctx; \pctx \ent[F] e \T A \while \cstr
  }{
    \ctx; \pctx \ent[F, \alpha, \drt] \absurd e \T \alpha \E \drt \while \cstr, A \le \emptyty
  }
  
  \inferrule[Cstr-App]{
    \ctx; \pctx \ent[F_1] e_1 \T A_1 \while \cstr_1 \\
    \ctx; \pctx \ent[F_2] e_2 \T A_2 \while \cstr_2
  }{
    \ctx; \pctx \ent[F_1, F_2, \alpha, \drt] e_1 \, e_2 \T \alpha \E \drt
      \while \cstr_1, \cstr_2, A_1 \le (A_2 \to[\drt] \alpha)
  }

  \inferrule[Cstr-Val]{
    \ctx; \pctx \ent[F] e \T A \while \cstr
  }{
    \ctx; \pctx \ent[F, \drt] \val e \T A \E \drt \while \cstr
  }

  \inferrule[Cstr-Op]{
    \ctx; \pctx \ent[F_1] e_1 \T A_1 \while \cstr_1 \\
    \op \T A^\op \to B^\op \in \sig(E) \\
    \ctx; \pctx \ent[F_2] e_2 \T A_2 \while \cstr_2 \\
    \ctx, y \T B^\op; \pctx \ent[F] c \T A \E \Drt \while \cstr
  }{
    \ctx; \pctx \ent[F_1, F_2, F, \rgn, \rgn*, \drt] \call{e_1}{\op}{e_2}{\cont{y}{c}} \T A \E \set{\op \T \rgn \mid \drt}
      \while \cstr_1, \cstr_2, \cstr, A_1 \le E^{\rgn*}, A_2 \le A^\op, \rgn* \le \rgn, \Drt \le \set{\op \T \rgn \mid \drt}
  }

  \inferrule[Cstr-Let]{
    \ctx; \pctx \ent[F_1] c_1 \T A \E \Drt_1 \while \cstr_1 \\
    \ctx, x \T \alpha; \pctx \ent[F_2] c_2 \T B \E \Drt_2 \while \cstr_2
  }{
    \ctx; \pctx \ent[F_1, F_2, \alpha, \drt] \letin{x = c_1} c_2 \T B \E \drt
      \while \cstr_1, \cstr_2, A \le \alpha, \Drt_1 \le \drt, \Drt_2 \le \drt
  }

  \inferrule[Cstr-LetVal]{
    \ctx; \pctx \ent[F_1] e \T A \while \cstr_1 \\
    \ctx; \pctx, (x \T \fra{F_1} A \mid \cstr_1) \ent[F_2] c \T \C \while \cstr_2
  }{
    \ctx; \pctx \ent[F_2] \letvalin{x = e} c \T \C \while \cstr_2
  }

 \inferrule[Cstr-LetRec]{
    \ctx, f \T \alpha_1 \to[\drt] \alpha_2, x \T \alpha_1; \pctx \ent[F_1] c_1 \T \C \while \cstr_1 \\
    \ctx, f \T \alpha_1 \to[\drt] \alpha_2; \pctx \ent[F_2] c_2 \T \D \while \cstr_2
  }{
    \ctx; \pctx \ent[F_1, F_2, \alpha_1, \alpha_2, \drt] \letrecin{f \, x = c_1} c_2 \T \D \while \cstr_1, \cstr_2, \C \le \alpha_2 \E \drt
  }

  \inferrule[Cstr-With]{
    \ctx; \pctx \ent[F_1] e \T A \while \cstr_1 \\
    \ctx; \pctx \ent[F_2] c \T \C \while \cstr_2
  }{
    \ctx; \pctx \ent[F_1, F_2, \alpha, \drt] \withhandle{e}{c} \T \alpha \E \drt
      \while \cstr_1, \cstr_2, A \le (\C \hto \alpha \E \drt)
  }
\end{mathpar}
\caption{
  The inductive definition of the inference judgements
  $\ctx; \pctx \ent[F] e \T A \while \cstr$ and $\ctx; \pctx \ent[F] c \T \C \while \cstr$.
  The rule for handlers is given in the main text.}
\label{fig:constraints}
\hrulefill
\end{figure}

As announced at the beginning of Section~\ref{sub:parametric-types},
regions and dirts have a fixed representation with parameters.
Thus in \rulename{Cstr-Inst}, we assign each instance a fresh region parameter and add a suitable constraint.
Similarly, we cannot simply state that the dirt of $\kord{val}$ is empty.
Instead, in \rulename{Cstr-Val}, we assign it a fresh dirt parameter~$\drt$
that needs to satisfy no constraints.
That means that we may replace $\drt$ by anything, including the empty set.

Though we get an equivalent set of constraints in the rule \rulename{Cstr-Op}
if we use a single region parameter,
we introduce two for technical reasons, discussed in Section~\ref{sub:simplifying-region-constraints}.
The rule~\rulename{Cstr-With} is analogous to~\rulename{Cstr-App}.

Otherwise, the rules for the standard constructs are
similar to ones in the Hindley-Milner algorithm~\cite[p.~322]{pierce2002types},
except that we need to use (correctly oriented) inequalities instead of equalities in the constraints.
In \rulename{Cst-LetVal} we can safely generalize over all the fresh parameters $F_1$ generated
while inferring the type of $e$ because they are guaranteed to be distinct from any parameters appearing in $\ctx$.

This leaves us with
\[
  \inferrule[Cstr-Hand]{
    \ctx, x \T \alpha_\text{in}; \pctx \ent[F_v] c_v \T \D_v \while \cstr_v \\
    (\prms_i)_i
  }{
    \ctx; \pctx \ent[F] (
      \handler
      \val x \mapsto c_v \case
      (\call{e_i}{\op_i}{x}{k} \mapsto c_i)_i
    ) \T \C \hto \D
      \while \cstr
  }
\]
To start, we take type parameters $\alpha_\text{in}$ and $\alpha_\text{out}$ to represent the incoming and outgoing type of the handler.
Next, we take $\ops$ to be the set of all distinct operation symbols listed in operation cases.
For each $\op \in \ops$, we take fresh parameters $\rgn_\text{in}^\op$ and $\rgn_\text{out}^\op$
that represent the region assigned to $\op$ in the incoming and outgoing dirt of the handler.
Finally, we take fresh parameters $\drt_\text{in}^\ops$ and $\drt_\text{out}^\ops$ to represent the rest of incoming and outgoing dirt.
The incoming and outgoing types are then
\[
  \C \defeq \alpha_\text{in} \E \set{(\op \T \rgn_\text{in}^\op)_{\op \in \ops} \mid \drt_\text{in}}
  \qquad\text{and}\qquad
  \D \defeq \alpha_\text{out} \E \set{(\op \T \rgn_\text{out}^\op)_{\op \in \ops} \mid \drt_\text{out}}
\]
After introducing the necessary parameters, we infer the type and constraints of the value case.
Next, for each operation case, in the premises~$\prms_i$, consisting of:
\[
  \ctx \ent[F_i] e_i \T A_i \while \cstr_i \qquad
  \op_i \T A^{\op_i} \to B^{\op_i} \in \sig(E_i) \qquad
  \ctx, x \T A^{\op_i}, k \T B^{\op_i} \to \D \ent[F_i'] c_i \T \D_i \while \cstr_i'
\]
we check the suitability of operation,
and infer the types and constraints of the handled instance~$e_i$ and of the operation case~$c_i$.

We end up with the a set of constraints~$\cstr$ consisting of the following five parts:
\begin{itemize}
\label{pag:hand-rules}
\item
  constraints~$\cstr_v$ inherited from the value case and
  constraints~$\cstr_i$ and $\cstr_i'$ inherited from each operation case;
\item
  constraint $\D_v \le \D$ stating that the outgoing type~$\D$ subsumes the type of the value case
  and constraints $\D_i \le \D$ stating the same for all the operation cases;
\item
  constraints $A_i \le E^{\rgn*_i}$ stating that the type~$A_i$ of the instance expression~$e_i$ is subsumed by the effect type~$E^{\rgn*_i}$ for some fresh $\rgn*_i$;
\item
  constraints $\rgn_\text{in}^\op \le \rgn_\text{out}^\op \cup \uniq{\op = \op_i}{\rgn*_i}$ for each $\op \in \ops$ ---
  the outgoing dirt must be big enough to cover all operations in the incoming dirt
  that are not surely handled by one of the operation cases; and
\item
  a constraint $\drt_\text{in} \le \drt_\text{out}$ --- any operation that is not listed in a handler cannot be handled,
  so must appear in the outgoing dirt as well.
\end{itemize}
The set~$F$ of all fresh parameters gathers all fresh parameters mentioned above and equals
\[
  F \defeq
    F_v, (F_i)_i, (F_i')_i,
    \alpha_\text{in}, \alpha_\text{out}, \drt_\text{in}, \drt_\text{out}, (\rgn*_i)_i,
    (\rgn_\text{in}^\op)_{\op \in \ops}, (\rgn_\text{out}^\op)_{\op \in \ops}
\]

Looking at the presented inference rules,
we see that there is exactly one rule that applies to each language construct,
so we can assign a unique type and set of constraints to each term (up to a renaming of fresh parameters).
This allows us to turn the rules into a recursive function,
which computes exactly the information about all the types we can assign to a given term:

\begin{thm}[Soundness \& completeness]
\label{thm:completeness}
Let $\ctx$ be a closed context.
\begin{itemize}
\item
  If we have $\ctx \ent e \T A'$ for some $A'$,
  then $\ctx \ent[F] e \T A \while \cstr$ holds and
  \[
    \types{\cstr} = \set{ A'' \mid (\ctx \ent e \T A'') }
  \]
\item
  We have $\ctx \ent[F] c \T \C \while \cstr$
  if and only if $\ctx \ent c \T \C'$ holds for some $\C'$.
  In this case
  \[
    \types[\C]{\cstr} = \set{ \C'' \mid (\ctx \ent c \T \C'') }
  \]
\end{itemize}
\end{thm}

\begin{rem}
\label{rem:glitch}
We limit the parameter and result types in the signature~$\sig$ to basic types
because $\sig$ is shared between typing judgements, which feature concrete regions and dirt,
and inference judgements, which represent regions and dirt exclusively with parameters.
There are three ways of reconciling this conflict:
\begin{itemize}
\item
  Extend the language of constraints with concrete upper bounds of the form $\rgn \le \Rgn$ and $\drt \le \emptyset$.
  Then we may, say, replace any occurrence of $A^\op = E^{\set{\inst_1, \inst_2}}$ in inference judgements
  with $E^{\rgn*}$ for some fresh $\rgn*$,
  and add constraints $\inst_1 \in \rgn*, \inst_2 \in \rgn*, \rgn* \le \set{\inst_1, \inst_2}$,
  or replace $B^\op = \unitty \to[\hash{\inst}{\op}] \unitty$
  with a suitably fresh $\unitty \to[\op \T \rgn \mid \drt] \unitty$ and
  constraints $\inst \in \rgn, \rgn \le \set{\inst}, \drt \le \emptyset$.
\item
  Extend monomorphic types with \emph{wild card} regions and dirt.
  For example, the type $\kord{exception}^\wild$ would capture any exception,
  no matter which concrete region it comes from.
  Similarly, the dirt $\hash{\wild}{\kord{raise}}$ would mean that a computation raises some exception,
  though we do not know which one,
  while the dirt $\wild$ would mean that any operation may get called.

  This solution agrees with practice~\cite{bauer2012programming},
  where most of types that appear in the signature are already basic,
  and the only two deviations so far are cooperative multithreading and delimited continuations,
  which both take functions as parameters.
  However, in both cases, we are not interested in imposing any limits on this dirt,
  so a wild card dirt would fit our goal.

  To add wild card regions and dirt to core \Eff, we need to:
  (1) add subtyping rules such as $R \subseteq \wild$,
  (2) extend \rulename{Op} with a condition that if $e_1 \T E^\wild$ then $\hash{\wild}{\op} \in \Drt$,
  and (3) adapt the rule \rulename{Hand} to ensure that any wild card dirt cannot be handled.
\item
  Each of the above approaches has its advantages,
  and in addition, the two are compatible,
  so one may consider both in a practical implementation.
  However,
  the first approach leads a more powerful \emph{prescriptive} effect system
  which is well beyond the scope of this paper
  (discussed more in the Conclusion),
  while the second one is routine but messy.
  Thus, we opt for the simplest option:
  prohibit any types that include regions and dirt from appearing in $\sig$.
\end{itemize}
\end{rem}

\section{Unifying constraints}
\label{sec:unifying}

Unfortunately, unlike in ML,
subtyping prevents us from computing a principal type from a given set of constraints~\cite{pottier1998type}.
For example (ignoring dirt for a moment), if we just drop the constraint in the type
\[
  (\alpha \to \beta) \to (\alpha \to \beta) \while \beta \le \alpha
\]
we get a type that captures too many closed types.
On the other hand, the more restricted parametric type $(\gamma \to \gamma) \to (\gamma \to \gamma)$ is
too strict because it fails to capture the type
\[
  (E^{\set{\inst, \inst'}} \to E^{\set{\inst}}) \to (E^{\set{\inst, \inst'}} \to E^{\set{\inst}})
\]
or any of its subtypes (otherwise, the subtyping rules would imply $E^{\set{\inst, \inst'}} \le \gamma \le E^{\set{\inst}}$).

Instead, the best we can do is to simplify the constraints as much as possible.
First, we are going to reduce the constraints down to a more convenient and basic form.
Constraints in this form always admit a solution,
so the reduction also detects any unsolvable constraints.

\begin{defi}
A set of constraints~$\cstr$ is \emph{unified},
if all constraints are \emph{decomposed} down to ones between parameters
and the set is \emph{closed} under logical implication.
In particular, $\cstr$ may contain only constraints of the form
\[
  \alpha \le \alpha', \qquad
  \rgn \le \rgn' \cup \uniq{i}{\rgn*_i}, \qquad
  \inst \in \rgn \cup \uniq{i}{\rgn*_i}, \qquad
  \drt \le \drt',
\]
and the following closure properties must hold:
\begin{itemize}
\item
  if $(\alpha_1 \le \alpha_2) \in \cstr$ and $(\alpha_2 \le \alpha_3) \in \cstr$,
  then $(\alpha_1 \le \alpha_3) \in \cstr$;
\item
  if $(\rgn_1 \le \rgn_2 \cup \uniq{i \in I} \rgn*_i) \in \cstr$
  and $(\rgn_2 \le \rgn_3 \cup \uniq{i \in J} \rgn*_i) \in \cstr$,
  then $(\rgn_1 \le \rgn_3 \cup \uniq{i \in I \cup J} \rgn*_i) \in \cstr$;
\item
  if $(\inst \in \rgn_1 \cup \uniq{i \in I} \rgn*_i) \in \cstr$
  and $(\rgn_1 \le \rgn_2 \cup \uniq{i \in J} \rgn*_i) \in \cstr$,
  then $(\inst \in \rgn_2 \cup \uniq{i \in I \cup J} \rgn*_i) \in \cstr$;
\item
  if $(\drt_1 \le \drt_2) \in \cstr$ and $(\drt_2 \le \drt_3) \in \cstr$,
  then $(\drt_1 \le \drt_3) \in \cstr$.
\end{itemize}
To avoid circular types, we need to track all type parameters of the same shape.
So, we assume that $\cstr$ is equipped with an equivalence relation~$\approx_\cstr$
on type parameters, such that $(\alpha \le \alpha') \in \cstr$ implies $\alpha \approx_\cstr \alpha'$.
For solutions of a unified set of constraints~$\cstr$, we consider only such $\sol \models \cstr$,
for which we also have $\sol(\alpha) \approx \sol(\alpha')$ for any $\alpha \approx_\cstr \alpha'$.
\end{defi}

\begin{lem}
\label{lem:unify}
If a set of constraints~$\cstr$ is unified, there exists a solution~$\sol \models \cstr$.
\end{lem}

In Figure~\ref{fig:unify}, we define an algorithm $\unify$,
which is similar to Robinson's unification algorithm~\cite[p.~327]{pierce2002types},
except that it returns a set of unified constraints in addition to the unifying substitution.
The algorithm is defined recursively, passing around a triple $(\sol; \cstr; \queue)$, where
$\sol$ is a unifying substitution of replaced parameters (initially taken to be the identity),
$\cstr$ is a set of already unified constraints (initially $\cstr$ is empty while $\approx_\cstr$ is the identity relation on all type parameters in $\queue$), and
the queue $\queue$ is a set of constraints yet to be processed.

The unifying substitution is not a closed substitution (Definition~\ref{defi:closed-substitution}),
which maps type parameters to closed types, etc., but one that maps them to parametric ones. 
\begin{defi}
A \emph{substitution}~$\sol$ is a mapping that maps:
  each type parameter~$\alpha$ to a type~$\sol(\alpha)$,
  each region parameter~$\rgn$ to a region parameter~$\sol(\rgn)$,
  each inhabited region parameter~$\rgn*$ to an inhabited region parameter~$\sol(\rgn*)$,
  and each dirt parameter~$\drt^\ops$ to a dirt~$\sol(\drt^\ops)$
  of the form $\set{(\op \T \rgn_\op)_{\op \in \ops'} \mid \drt'^{\ops \cup \ops'}}$,
  where $\ops'$ is disjoint from $\ops$.
This ensures that the dirt $\sol(\drt^\ops)$
does not capture any operations from $\ops$.
\end{defi}
We write substitutions by listing a set of all the non-idempotent rules.
We can extend a substitution from parameters to other constructs just like we extended closed substitutions.
This allows us to compose substitutions and additionally,
to compose a closed substitution~$\sol$ with a substitution~$\sol'$
and obtain a closed substitution~$\sol \circ \sol'$.

\begin{figure}
\small
\hrulefill
\begin{align*}
  &\unify(\sol;\ \cstr;\ \cdot) =
    (\sol,\ \cstr) \\
  &\unify(\sol;\ \cstr;\ A \le A', \queue) = \\
  &\quad \kpre{match} A \le A' \kpost{with} \\
  &\quad \case A \le A \mapsto
    \unify(\sol;\ \cstr;\ \queue) \\
  &\quad \case \alpha \le \alpha' \mapsto
    \unify(\sol;\ \cstr \uplus (\alpha \le \alpha');\ \queue) \\
  &\quad \case \alpha \le A \mapsto
    \kpre{if} \exs{\alpha' \approx_\cstr \alpha} \alpha' \in \kord{free}(A) \kop{then} \kord{failure} \kpost{else} \\
  &\qquad \letin{\sol' = \set{\alpha' \mapsto \kord{refresh}(A) \mid \alpha' \approx_\cstr \alpha}} \\
  &\qquad \letin{\cstr' = \set{(\alpha' \le \alpha'') \in \cstr \mid \alpha \approx_\cstr \alpha' \approx_\cstr \alpha''}} \\
  &\qquad \unify(
    \sol' \circ \sol;\
    \cstr - \cstr';\ 
    \sol'(\queue), \sol'(\alpha) \le A, \sol'(\cstr')) \\
  &\quad \case A \le \alpha \mapsto
    \kpre{if} \exs{\alpha' \approx_\cstr \alpha} \alpha' \in \kord{free}(A) \kop{then} \kord{failure} \kpost{else} \\
  &\qquad \letin{\sol' = \set{\alpha' \mapsto \kord{refresh}(A) \mid \alpha' \approx_\cstr \alpha}} \\
  &\qquad \letin{\cstr' = \set{(\alpha' \le \alpha'') \in \cstr \mid \alpha \approx_\cstr \alpha' \approx_\cstr \alpha''}} \\
  &\qquad \unify(
    \sol' \circ \sol;\
    \cstr - \cstr';\ 
    \sol'(\queue), A \le \sol'(\alpha), \sol'(\cstr')) \\
  &\quad \case E^{\rgn*} \le E^{\rgn*'} \mapsto
    \unify(\sol;\ \cstr \uplus \rgn* \le \rgn*';\ \queue) \\
  &\quad \case (A \to \C) \le (A' \to \C') \mapsto
    \unify(\sol;\ \cstr;\ \queue, A' \le A, \C \le \C') \\
  &\quad \case (\C \hto \D) \le (\C' \hto \D') \mapsto
    \unify(\sol;\ \cstr;\ \queue, \C' \le \C, \D \le \D') \\
  &\quad \case \kord{otherwise} \mapsto \kord{failure} \\
  &\unify(\sol;\ \cstr;\ A \E \Drt \le A' \E \Drt', \queue) =
    \unify(\sol;\ \cstr;\ \queue, A \le A', \Drt \le \Drt') \\
  &\unify(\sol;\ \cstr;\ \rgn \le \rgn' \cup \uniq{i}{\rgn*_i}, \queue) =
    \unify(\sol;\ \cstr \uplus (\rgn \le \rgn' \cup \uniq{i}{\rgn*_i});\ \queue) \\
  &\unify(\sol;\ \cstr;\ \inst \in \rgn \cup \uniq{i}{\rgn*_i}, \queue) =
    \unify(\sol;\ \cstr \uplus (\inst \in \rgn \cup \uniq{i}{\rgn*_i});\ \queue) \\
  &\unify(\sol;\ \cstr;\ \set{(\op \T \rgn_\op)_{\op \in \ops} \mid \drt_1^\ops} \le \set{(\op \T \rgn_\op')_{\op \in \ops'} \mid \drt_2^{\ops'}}, \queue) = \\
  &\quad\kpre{if} \ops = \ops' \kop{then} \\
  &\qquad \unify(\sol;\ \cstr \uplus \set{\rgn_\op \le \rgn_\op' \mid \op \in \ops} \uplus (\drt_1 \le \drt_2),\ \queue) \\
  &\quad\kord{else} \\
  &\qquad \letin{\sol_1 = \kpre{if} \ops' \subset \ops \kop{then} \set{} \kop{else} \set{\drt_1 \mapsto \set{(\op \T \rgn_\op)_{\op \in \ops' - \ops} \mid \drt_1'^{\ops \cup \ops'}}}} \\
  &\qquad \letin{\sol_2 = \kpre{if} \ops \subset \ops' \kop{then} \set{} \kop{else} \set{\drt_2 \mapsto \set{(\op \T \rgn_\op')_{\op \in \ops - \ops'} \mid \drt_2'^{\ops \cup \ops'}}}} \\
  &\qquad \letin{\cstr' =
    \set{(\drt \le \drt') \in \cstr \mid \drt \sim_\cstr \drt' \sim_\cstr \drt_1} \cup
    \set{(\drt \le \drt') \in \cstr \mid \drt \sim_\cstr \drt' \sim_\cstr \drt_2}
  } \\
  &\qquad \unify(\sol_1 \circ \sol_2 \circ \sol;\ (\cstr - \cstr'),\ \sol'(\queue) \cup \sol'(\cstr'))
\end{align*}
\caption{Definition of the constraint unification algorithm~$\unify$.}
\label{fig:unify}
\hrulefill
\end{figure}

The unification works as follows.
Take, say, a constraint $\alpha_1 \le (\alpha_2 \to[\drt_1] \boolty)$.
Since the rules of structural subtyping admit only comparison between types of the same shape,
the only way to satisfy this constraint is to set $\alpha_1$ to be some function type into $\boolty$.
So, we decompose the constraint by replacing $\alpha_1$ with some fresh $\alpha_3 \to[\drt_2] \boolty$
and adding constraints $\alpha_2 \le \alpha_3$ and $\drt_2 \le \drt_1$

We add these constraints using the closure operator $\uplus$, defined in Figure~\ref{fig:closure},
which extends $\cstr$ with a given constraint and all constraints it implies.
This is done with a simplified version of an algorithm for computing the transitive closure of a graph~\cite{pottier1998type},
except that for region parameters, we also have to take instances and handled regions into account.

\begin{figure}[h]
\small
\hrulefill
\begin{align*}
  \cstr \uplus (\alpha_1 \le \alpha_2) =
    \cstr
    &\cup \set{\alpha_1' \le \alpha_2' \mid (\alpha_1' \le \alpha_1) \in \cstr, (\alpha_2 \le \alpha_2') \in \cstr} \\
  \cstr \uplus (\rgn_1 \le \rgn_2 \cup \varuniq{i \in I}{\rgn*_i}) =
    \cstr
    &\cup \set{\rgn_1' \le \rgn_2' \cup \varuniq{\hspace{-10em}i \in I \cup J \cup K\hspace{-10em}}{\rgn*_i} \mid
      (\rgn_1' \le \rgn_1 \cup \varuniq{i \in J}{\rgn*_i}) \in \cstr,
      (\rgn_2 \le \rgn_2' \cup \varuniq{i \in K}{\rgn*_i}) \in \cstr} \\
    &\cup \set{\inst \in \rgn_2' \cup \varuniq{\hspace{-10em}i \in I \cup J \cup K\hspace{-10em}}{\rgn*_i} \mid
      (\inst \le \rgn_1 \cup \varuniq{i \in J}{\rgn*_i}) \in \cstr,
      (\rgn_2 \le \rgn_2' \cup \varuniq{i \in K}{\rgn*_i}) \in \cstr} \\
  \cstr \uplus (\inst \in \rgn \cup \varuniq{i \in I}{\rgn*_i}) =
    \cstr
    &\cup \set{\inst \in \rgn' \cup \varuniq{\hspace{-10em}i \in I \cup J\hspace{-10em}}{\rgn*_i} \mid
      (\rgn \le \rgn' \cup \varuniq{i \in J}{\rgn*_i}) \in \cstr} \\
  \cstr \uplus (\drt_1 \le \drt_2) =
    \cstr
      &\cup \set{\drt_1' \le \drt_2' \mid (\drt_1' \le \drt_1) \in \cstr, (\drt_2 \le \drt_2') \in \cstr}
\end{align*}
\caption{Definition of the closure operator~$\uplus$.
In all cases, we assume that $\cstr$ implicitly contains all reflexive constraints,
so that, for example in the first case,
the added set also includes $\alpha_1 \le \alpha_2$ and all constraints of the form $\alpha_1 \le \alpha_2'$ and $\alpha_1' \le \alpha_2$.}
\label{fig:closure}
\hrulefill
\end{figure}

Before decomposition, we need to perform an \emph{occur check}
in order to prevent ill-formed types and ensure termination.
This check is slightly more involved than usually~\cite[p.~327]{pierce2002types}.
Say that we want to unify the set of constraints (let us again ignore dirt):
\[
  \set{(\alpha_1 \to \alpha_2) \le \alpha_3, \alpha_4 \le \alpha_1, \alpha_4 \le \alpha_3}
\]
We decompose $\alpha_3$ as some $\alpha_5 \to \alpha_6$ and end up with constraints
\[
  \set{\alpha_5 \le \alpha_1, \alpha_2 \le \alpha_6, \alpha_4 \le \alpha_1, \alpha_4 \le (\alpha_5 \to \alpha_6)}
\]
Now, we need to decompose $\alpha_4$, then $\alpha_1$, then $\alpha_5$ and the whole thing repeats.
So we need to check not only that $\alpha$ is not in $\kord{free}(A)$,
the set of all parameters that occur in $A$,
but also that no other parameters in its skeleton are.

When decomposing $\alpha \le A$,
we also replace each $\alpha' \approx_\cstr \alpha$ with a (distinct) fresh copy of $A$.
We do this by repeatedly calling a function $\kord{refresh}(A)$ that on each call returns a type of the same form as $A$,
except with all its type, region, and dirt parameters replaced with fresh ones.
Because of this expansion, any constraints~$\cstr' \subseteq \cstr$
that mention parameters from the skeleton of $\alpha$ are no longer decomposed.
So we need to take them out of the otherwise unified $\cstr$ and put them back into the queue~$\queue$.

On the remaining unified set of constraints $\cstr - \cstr'$,
we define $\approx_{\cstr - \cstr'}$ to be as before,
except that we remove the whole skeleton of $\alpha$,
and add the freshly generated parameters into the skeletons of matching parameters in $A$.
For example, if the skeletons of $\approx_\cstr$ were
\[
  \set{\alpha_1, \alpha_2}, \set{\alpha_3, \alpha_4, \alpha_5}, \set{\alpha_6}
\]
and we decompose the constraint $\alpha_1 \le \alpha_3 \to[\drt_1] \alpha_6$,
we replace $\alpha_1$ by $\alpha_7 \to[\drt_2] \alpha_8$ and $\alpha_2$ by $\alpha_9 \to[\drt_3] \alpha_{10}$,
and the skeletons of $\approx_{\cstr - \cstr'}$ are
\[
  \set{\alpha_3, \alpha_4, \alpha_5, \alpha_7, \alpha_9}, \set{\alpha_6, \alpha_8, \alpha_{10}}
\]
When we unify a constraint $\alpha \le \alpha'$, we merge the skeletons of $\alpha$ and $\alpha'$.

For dirt constraints,
we similarly expand both sides so to list the same operations.
For example, if we have a constraint
$\set{\op \T \rgn_1 \mid \drt_1^{\op}} \le \set{\op' \T \rgn_2 \mid \drt_2^{\op'}}$,
we replace $\drt_1$ with some fresh $\set{\op' \T \rgn_3 \mid \drt_3^{\op, \op'}}$
and $\drt_2$ with $\set{\op \T \rgn_4 \mid \drt_4^{\op, \op'}}$
and add constraints $\rgn_1 \le \rgn_4$, $\rgn_3 \le \rgn_2$, and $\drt_3 \le \drt_4$.
We define $\sim_\cstr$ to be the equivalence relation on dirt parameters generated by $(\drt \le \drt') \in \cstr$,
so that we can capture all related dirt parameters and expand them at the same time.

An example of a full run of the unification algorithm is given in Figure~\ref{fig:unification}.

\begin{figure}[h]
\small
\hrulefill
\newcommand{\bigset}[1]{\Bigl\{#1\Big\}}
\newcommand{\bigunify}[4][\quad]{\unify\Big(\bigset{#2};#1\bigset{#3};\quad\bigset{#4}\Big)}
\begin{align*}
  \bigunify
    {&}
    {}
    {(\alpha_1 \to[\drt_1] \natty) \le \alpha_2,\ \set{\op \T \rgn*_1 \mid \drt_2} \le \drt_1)}
  = \\
  \bigunify
    {&\alpha_2 \mapsto (\alpha_3 \to[\drt_3] \natty)}
    {}
    {(\alpha_1 \to[\drt_1] \natty) \le (\alpha_3 \to[\drt_3] \natty),\ \set{\op \T \rgn*_1 \mid \drt_2} \le \drt_1}
  = \\
  \bigunify
    {&\alpha_2 \mapsto (\alpha_3 \to[\drt_3] \natty)}
    {}
    {\alpha_3 \le \alpha_1,\ (\natty \E \drt_1) \le (\natty \E \drt_3),\ \set{\op \T \rgn*_1 \mid \drt_2} \le \drt_1}
  = \\
  \bigunify
    {&\alpha_2 \mapsto (\alpha_3 \to[\drt_3] \natty)}
    {\alpha_3 \le \alpha_1}
    {(\natty \E \drt_1) \le (\natty \E \drt_3),\ \set{\op \T \rgn*_1 \mid \drt_2} \le \drt_1}
  = \\
  \bigunify
    {&\alpha_2 \mapsto (\alpha_3 \to[\drt_3] \natty)}
    {\alpha_3 \le \alpha_1}
    {\natty \le \natty,\ \drt_1 \le \drt_3,\ \set{\op \T \rgn*_1 \mid \drt_2} \le \drt_1}
  = \\
  \bigunify
    {&\alpha_2 \mapsto (\alpha_3 \to[\drt_3] \natty)}
    {\alpha_3 \le \alpha_1}
    {\drt_1 \le \drt_3,\ \set{\op \T \rgn*_1 \mid \drt_2} \le \drt_1}
  = \\
  \bigunify
    {&\alpha_2 \mapsto (\alpha_3 \to[\drt_3] \natty)}
    {\alpha_3 \le \alpha_1,\ \drt_1 \le \drt_3}
    {\set{\op \T \rgn*_1 \mid \drt_2} \le \drt_1}
  = \\
  \bigunify[\\]
    {&\alpha_2 \mapsto (\alpha_3 \to[\op \T \rgn*_3 \mid \drt_5] \natty),\ 
      \drt_1 \mapsto \set{\op \T \rgn*_2 \mid \drt_4},\ 
      \drt_3 \mapsto \set{\op \T \rgn*_3 \mid \drt_5}}
    {&\alpha_3 \le \alpha_1}
    {\set{\op \T \rgn*_1 \mid \drt_2} \le \set{\op \T \rgn*_2 \mid \drt_4},\ 
      \set{\op \T \rgn*_2 \mid \drt_4} \le \set{\op \T \rgn*_3 \mid \drt_5}}
  = \\
  \bigunify[\\]
    {&\alpha_2 \mapsto (\alpha_3 \to[\op \T \rgn*_3 \mid \drt_5] \natty),\ 
      \drt_1 \mapsto \set{\op \T \rgn*_2 \mid \drt_4},\ 
      \drt_3 \mapsto \set{\op \T \rgn*_3 \mid \drt_5}}
    {&\alpha_3 \le \alpha_1,\ \rgn*_1 \le \rgn*_2,\ \drt_2 \le \drt_4}
    {\set{\op \T \rgn*_2 \mid \drt_4} \le \set{\op \T \rgn*_3 \mid \drt_5}}
  = \\
  \bigunify[\\]
    {&\alpha_2 \mapsto (\alpha_3 \to[\op \T \rgn*_3 \mid \drt_5] \natty),\ 
      \drt_1 \mapsto \set{\op \T \rgn*_2 \mid \drt_4},\ 
      \drt_3 \mapsto \set{\op \T \rgn*_3 \mid \drt_5}}
    {&\alpha_3 \le \alpha_1,\ \rgn*_1 \le \rgn*_2 \le \rgn*_3,\ \drt_2 \le \drt_4 \le \drt_5}
    {}
  = \\
  \Bigl(\bigset{
    &\alpha_2 \mapsto (\alpha_3 \to[\op \T \rgn*_3 \mid \drt_5] \natty),\ 
    \drt_1 \mapsto \set{\op \T \rgn*_2 \mid \drt_4},\ 
    \drt_3 \mapsto \set{\op \T \rgn*_3 \mid \drt_5}}, \\
   \bigset{
    &\alpha_3 \le \alpha_1,\ \rgn*_1 \le \rgn*_2 \le \rgn*_3,\ \drt_2 \le \drt_4 \le \drt_5
    }\Bigr)
\end{align*}
\caption{The unification of the set of constraints $\set{(\alpha_1 \to[\drt_1] \natty) \le \alpha_2, \set{\op \T \dot{\rho}_2 \mid \drt_2} \le \drt_1}$.
We display constraints in a coalesced form.
For example, we write $\dot\rho_1 \le \dot\rho_2 \le \dot\rho_3$ instead of $\dot\rho_1 \le \dot\rho_2$, $\dot\rho_2 \le \dot\rho_3$ and $\dot\rho_1 \le \dot\rho_3$.}
\label{fig:unification}
\hrulefill
\end{figure}

\begin{prop}
\label{prop:unify}
For any set of constraints~$\cstr$,
the algorithm $\unify(\emptyset; \emptyset; \cstr)$ always halts.
\begin{itemize}
\item
  If it halts with a failure, then $\cstr$ has no solutions.
\item
  If it halts returning $(\sol', \cstr')$, then $\cstr$ has a solution.
  Furthermore, solutions $\sol \models \cstr$ are exactly all
  the ones of the form $\sol = \sol'' \circ \sol'$,
  where $\sol'' \models \cstr'$.
\end{itemize}
\end{prop}

\noindent Note that $\unify$ can fail for two reasons:
either we detect a cyclic constraint during occur check,
or we try to unify two incompatible types ---
no failure can happen due to unification of dirt and region constraints.
These are exactly the cases in which the Hindley-Milner algorithm fails~\cite[p.~327]{pierce2002types},
so our algorithm indeed infers the usual ML types,
except annotated with information about effects.

\begin{cor}
\label{cor:unify}
The two-way inference rules
\begin{mathpar}
  \mprset{fraction={===}}
  \inferrule[Unify-Expr]{
    \ctx \ent[F] e \T A \while \cstr
  }{
    \sol(\ctx) \ent[F] e \T \sol(A) \while \cstr'
  }

  \inferrule[Unify-Comp]{
    \ctx \ent[F] c \T \C \while \cstr
  }{
    \sol(\ctx) \ent[F] c \T \sol(\C) \while \cstr'
  }
\end{mathpar}
where $\unify(\cstr) = (\sol, \cstr')$, are sound.
\end{cor}

Since \rulename{Unify-Expr} and \rulename{Unify-Comp} are two-way rules,
we can be sure that no information is lost,
so we can perform unification phase not only after,
but also while gathering the constraints.
In fact, in \Eff, constraints are always kept in unified form,
so unification is performed each time we add a new constraint.
Though this strategy seems expensive, it offers many advantages:
\begin{enumerate}
\item
  Any ill-typed terms are caught as soon as possible.
  This does not influence the efficiency too much,
  but gives much more informative error messages.
\item
  In each rule, the inferred constraints consist mostly of constraints inherited from subterms.
  In \Eff, the inference algorithm uses a technique called \emph{$\lambda$-lifting}~\cite{pottier1998type}
  to ensure that subderivations share no common parameters.
  Then, if constraints $\cstr_1$ and $\cstr_2$, say, are unified,
  so is their union $\cstr_1 \cup \cstr_2$,
  and we need to perform closure only for the small number of additional constraints, specific to the current rule.
\item
  As we shall soon see, unified constraints may be garbage collected,
  drastically reducing their size and speeding up the algorithm.
\end{enumerate}

\section{Simplifying constraints}
\label{sec:simplifying}

\subsection{Garbage collection}
\label{sub:garbage-collection}

The main simplification technique we employ is
\emph{garbage collection}~\cite{pottier2001simplifying, simonet2003type, trifonov1996subtyping}.
We first recap the existing idea for type parameters,
and then extend it to dirt and region parameters so to fit into our setting.

Recall that a parametric type $A$ together with a unified set of constraints $\cstr$
captures exactly the closed types in the set $\types{\cstr}$.
Can we obtain a smaller set of constraints $\cstr'$ but keep $\types{\cstr'} = \types{\cstr}$?

Mark type parameters in $A$ as \emph{positive} or \emph{negative},
if they appear in a covariant or contravariant position, respectively.
For example, in $(\alpha_1 \to \alpha_2) \to \alpha_3$,
the parameters $\alpha_1$ and $\alpha_3$ are positive while $\alpha_2$ is negative.
It turns out that in $\cstr'$,
we only need to keep the constraints of the form $\alpha^- \le \alpha^+$,
where $\alpha^-$ is a negative and $\alpha^+$ is a positive.

Since $\cstr'$ contains less constraints than $\cstr$,
any solution of $\cstr'$ is a solution of $\cstr$ as well,
so we get $\types{\cstr} \subseteq \types{\cstr'}$.
For the other direction, we need to show that for any solution $\sol' \models \cstr'$,
there exists some $\sol \models \cstr$ such that $\sol(\alpha^+) \le \sol'(\alpha^+)$ holds for all positive~$\alpha^+$,
and $\sol'(\alpha^-) \le \sol(\alpha^-)$ holds for all negative $\alpha^-$.
It is then easy to see that $\sol(A) \le \sol'(A)$ holds as well,
and so any type $A'$ such that $\sol'(A) \le A'$ already appears in $\types{\cstr}$.
For exact details, see the proof of Proposition~\ref{prop:garbage-collection}.

\begin{defi}
The sets $\pos(A)$ of positive and $\neg(A)$ negative parameters in a given type~$A$ are defined as:
\begin{align*}
  \pos(\alpha) &= \set{\alpha} &
  \neg(\alpha) &= \emptyset \\
  \pos(A \to \C) &= \neg(A) \cup \pos(\C) &
  \neg(A \to \C) &= \pos(A) \cup \neg(\C) \\
  \pos(E^{\rgn}) &= \set{\rgn} &
  \neg(E^{\rgn}) &= \emptyset \\
  \pos(\C \hto \D) &= \neg(\C) \cup \pos(\D) &
  \neg(\C \hto \D) &= \pos(\C) \cup \neg(\D)
\end{align*}
The sets of both positive and negative parameters
in ground types ($\boolty$, $\natty$, \dots) are empty.
For dirty types, the sets of positive and negative parameters are defined as:
\begin{align*}
  \pos(A \E \set{\op_1 \T \rgn_1, \dots, \op_n \T \rgn_n \mid \drt}) &= \pos(A) \cup \set{\rgn_1, \dots, \rgn_n, \drt} \\
  \neg(A \E \Drt) &= \neg(A)
\end{align*}
\end{defi}

\begin{defi}
For a unified set of constraints~$\cstr$ and sets of parameters $P$ and $N$,
we define the \emph{garbage collected} constraints~$\gc(\cstr)$ as:
\begin{align*}
  \gc(\cstr) = {}
    &\set{(\alpha^- \le \alpha^+) \in \cstr \mid \alpha^- \in N, \alpha^+ \in P} \cup {} \\
    &\set{(\rgn^- \le \rgn^+ \cup \uniq{i}{\rgn*_i}) \in \cstr \mid \rgn^- \in N, \rgn^+ \in P} \cup {} \\
    &\set{(\inst \in \rgn^+ \cup \uniq{i}{\rgn*_i}) \in \cstr \mid \rgn^+ \in P} \cup {} \\
    &\set{(\drt^- \le \drt^+) \in \cstr \mid \drt^- \in N, \drt^+ \in P}
\end{align*}
On $\gc(\cstr)$, we define $\approx_{\gc(\cstr)}$ to be the restriction of $\approx_\cstr$ to the set $P \cup N$.
\end{defi}

To perform garbage collection on constraints in a typing judgement $\ctx \ent e \T A \while \cstr$,
we define $P$ to contain not just $\pos(A)$, but also $\neg(A_i)$ for all $(x_i \T A_i) \in \ctx$.
Conversely, $N$ must contain $\neg(A)$ and all $\pos(A_i)$.
We similarly extend $P$ with all region parameters $\rgn*_i$ in constraints of the form
$\rgn \le \rgn' \cup \uniq{i}{\rgn*_i}$ and $\inst \in \rgn \cup \uniq{i}{\rgn*_i}$.
Otherwise, garbage collection may drop some of their lower bounds and thus relax the constraints.
Recall that regions $\sol(\rgn*_i)$ are always non-empty, so decreasing them can only increase $\uniq{i}{\rgn*_i}$.

\begin{prop}
\label{prop:garbage-collection}
If the set of constraints~$\cstr$ is unified, the two-way inference rules
\begin{mathpar}
  \mprset{fraction={===}}
  \inferrule[GC-Expr]{
    \ctx \ent[F] e \T A \while \cstr
  }{
    \ctx \ent[F \cap (P \cup N)] e \T A \while \gc(\cstr)
  }

  \inferrule[GC-Comp]{
    \ctx \ent[F] c \T \C \while \cstr
  }{
    \ctx \ent[F \cap (P \cup N)] c \T \C \while \gc(\cstr)
  }
\end{mathpar}
where
\begin{align*}
  P_0 = {} &\pos(A) \cup \neg(\ctx) \\
  P = {}
      &P_0 \cup \set{\rgn*_i \mid (\rgn^- \le \rgn^+ \cup \uniq{i}{\rgn*_i}) \in \cstr, \rgn^- \in N, \rgn^+ \in P_0} \cup {} \\
      &\hphantom{P_0 \cup {}}\set{\rgn*_i \mid (\inst \in \rgn^+ \cup \uniq{i}{\rgn*_i}) \in \cstr, \rgn^+ \in P_0} \\
  N = {} &\neg(A) \cup \pos(\ctx)  
\end{align*}
are sound.

Furthermore, the set $\gc(\cstr)$ is unified.
\end{prop}

We present garbage collection in the form of two-way rules so that,
like unification, we can perform it while still gathering constraints.
This allows us to dispose any intermediate constraints as soon as possible,
making the whole algorithm very efficient.

Ignoring handled region parameters in region constraints,
the number of parameters is linear in the size of type,
and each constraint relates two parameters,
so the number of constraints is roughly quadratic in the size of the inferred type after garbage collection.
In practice, though, the number of constraints is often much smaller
and for typical functional programs,
the current implementation of the inference algorithm in \Eff is on a par with the one in OCaml
(more details can be found in Table~\ref{tab:benchmark}).
Without garbage collection, the inference algorithm experiences exponential blow-up and is unusable.

\begin{table}[h]
  \centering
  \begin{tabular}{rcc}
    \toprule
    filename                                & \Eff  & OCaml \\
    \midrule
    \texttt{garsia\_wachs.eff}/\texttt{.ml}  & 16 ms & 16 ms \\
    \texttt{list.eff}/\texttt{.ml}          & 46 ms & 42 ms \\ 
    \texttt{map.eff}/\texttt{.ml}           & 93 ms & 60 ms \\
    \texttt{set.eff}/\texttt{.ml}           & 53 ms & 33 ms \\ 
    \bottomrule
  \end{tabular}
  \caption{
    The table contains average ($N = 50, \sigma < 15\%$) file loading times
    for \Eff 3.1 and OCaml 4.01.0 on Mac OS X 10.9.3
    running on a 1.7 GHz Intel Core i5 with 4 GB of RAM.
    Except for a few initial definitions needed to bridge the differences,
    the \Eff and OCaml files are identical
    (they can be found in the \texttt{examples/benchmarks/} folder of the \Eff distribution).
    The test files contain only definitions and no executable code,
    so most of the loading time is spent on type-checking.
    To compensate for the time needed to start the interpreter and load the standard library,
    we subtracted the average time spent to load a blank file from all the entries (23 ms for \Eff and 12 ms for OCaml).
  }
  \label{tab:benchmark}
\end{table}

\subsection{Simplifying region constraints}
\label{sub:simplifying-region-constraints}

Garbage collection is extremely efficient,
but there is a small number of additional tactics we can offer
that simplify region constraints.
These are not meant to make the algorithm fast, but to make its output simpler.
For as we shall see in Section~\ref{sub:displaying-region-parameters},
region constraints tend to be the most difficult to present succinctly,
so it is crucial that they are as simple as possible before we display them.
The main idea behind all optimizations is that in region constraints,
handled region parameters~$\rgn*_i$ contribute to
the right-hand side $\rgn \cup \uniq{i}{\rgn*_i}$ only when they denote a singleton.

When determining if an inhabited region parameter $\rgn*$ denotes a singleton,
it is helpful to see that $\rgn*$ may appear as a covering region only in constraints of the form
$\rgn \le \rgn*$ and $\inst \in \rgn*$,
so ones with no handled regions $\uniq{i}{\rgn*_i}$.
To see this, one needs to laboriously check that no other constraints are added during inference, unification, or garbage collection. 
Adding an extra region parameter to \rulename{Cstr-Op} was a part of that effort.

Unlike garbage collection, each tactic is very basic, so we only sketch how they work.
\begin{itemize}
\item
Any parameter $\rgn*$ for which we have
both $\inst \in \rgn*$ and $\inst' \in \rgn*$ for some $\inst \ne \inst'$
cannot denote a singleton.
Thus, it may be removed from all singleton unions $\uniq{i}{\rgn*_i}$ in which it appears.
Similarly, if we have $\inst \in \rgn*$,
we may remove any occurrence of $\rgn*$ as a handled region in
constraints of the form $\inst' \in \rgn \cup \uniq{i}{\rgn*_i}$
as it cannot contain $\inst'$ and be a singleton at the same time.
\item
Next, take inhabited region parameters $\rgn*_1 \le \rgn*_2$
(recall that all inequalities between inhabited region parameters are of that form)
that both appear in some $\uniq{i}{\rgn*_i}$.
If $\rgn*_1$ does not denote a singleton, neither does $\rgn*_2$,
and if $\rgn*_1$ does denote a singleton,
$\rgn*_2$ can make no further contribution to $\uniq{i}{\rgn*_i}$.
So, we may safely remove $\rgn*_2$ in both cases.
\item
We may reduce the number of constraints by observing that if $I \subseteq J$,
the constraint
$\rgn \le \rgn' \cup \uniq{i \in I}{\rgn*_i}$
implies
$\rgn \le \rgn' \cup \uniq{i \in J}{\rgn*_i}$
and we may safely throw the latter one away.
\item
After all the simplifications,
we may further reduce the number of constraints
by another round of garbage collection,
in which we remove all lower bounds on region parameters
that no longer occur in singleton unions.
\end{itemize}

\subsection{Further simplifications}
\label{sub:further-simplifications}

Our algorithm also applies to the much simpler case when one
triggers effects with only operation symbols and no instances~\cite{kammar2013handlers}.
As this amounts to having a single instance~$\star$ of each effect,
we can drop all instance constraints $\inst \in \rgn \cup \uniq{i}{\rgn*_i}$
with at least one handled region parameter $\rgn*_i$.
This is because $\inst$ must be $\star$
and all inhabited region parameters $\rgn*$ must denote the singleton region $\set{\star}$,
so any such constraint is always satisfied.

Another tempting way to drop constraints is to drop all constraints
$\inst \in \rgn \cup \uniq{i}{\rgn*_i}$
for which $\inst \in \rgn*_i$ is the only lower bound for some $\rgn*_i$.
But, unlike other simplification techniques,
this may not be done until after we have gathered all the constraints,
because $\rgn*_i$ may still receive further lower bounds.

Similarly, we may want to improve the second tactic from Section~\ref{sub:simplifying-region-constraints},
and remove any handled parameters $\rgn*_2$ not only if we have $\rgn*_1 \le \rgn*_2$,
but also in case all lower bounds of $\rgn*_1$ are also lower bounds of $\rgn*_2$.
If $\rgn*_1 \le \rgn*_2$, this is implied by closure properties.
For example, if we have lower bounds $\rgn*_3 \le \rgn*_1$, $\rgn*_3 \le \rgn*_2$ and $\rgn*_4 \le \rgn*_2$,
we can always assign a smaller region to $\rgn*_1$ than to $\rgn*_2$.
Thus, if $\rgn*_1$ and $\rgn*_2$ appear in the same singleton union, we can safely drop $\rgn*_2$,
for if it denotes a singleton, so does $\rgn*_1$.
This condition is also stable under garbage collection ---
if $\rgn*_1$ is not negative, the constraint $\rgn*_1 \le \rgn*_2$ will be dropped,
but since $\rgn*_1$ and $\rgn*_2$ are both positive, all their lower bounds will be kept.
However, $\rgn*_1$ may similarly receive further lower bounds,
so we can perform the simplification only at the very end.

\section{Displaying inferred types}
\label{sec:displaying}

However, once the algorithm finishes,
we may perform further simplifications that lose information
but make the output much easier to understand.
These simplifications can be regarded as configurable,
so one may, if desired, omit some of them to reveal more details.
We emphasize that the presented techniques incur information loss
and should be used only for displaying output to a programmer.
For example, in \Eff's interactive loop, when we define a value,
\Eff shows the simplified type, but stores the full type in its typing environment.

\subsection{Displaying dirt parameters}
First, we can get rid of \emph{all} inequalities between dirt parameters.
Recall that increasing positive parameters always yields a valid typing,
so we are interested only in their smallest possible value.
Since this is exactly the union of all lower bounds,
we display a positive dirt parameter $\drt^+$ as the union $\bigcup_{\drt^- \le \drt^+} \drt^-$
of all smaller negative parameters.
So, given the constraints
\[
  \drt_1^- \le \drt_2^+, \drt_1^- \le \drt_3^+, \drt_4^- \le \drt_2^+
\]
we can replace $\drt_2^+$ with its lower bound $\drt_1^- \cup \drt_4^-$
and $\drt_3^+$ with its lower bound $\drt_1^-$.
In particular, if a positive dirt parameter $\drt$ has no lower bounds,
we write $A \E \drt$ as $A \E \emptyset$, and $A \to[\drt] B$ as $A \to B$.
We keep the negative parameters as they are,
because each usually captures a dirt of some input argument.

Since we can recover the constraints back from the displayed unions,
this technique loses no information,
though due to complex dirts,
the displayed types are no longer in a form we can use for inference.

\subsection{Displaying type parameters}
Next, we could use the same approach for representing type parameters.
However, though inequalities between type parameters are crucial for inference,
they do not offer much additional information to the programmer
besides telling what type parameters are of the same shape.
So, we can get rid of \emph{all} inequalities between type parameters,
and label all type parameters in the same skeleton with the same symbol.
So, instead of
\[
  (\alpha_1 \to[\drt] \alpha_2) \to \alpha_3 \to[\drt] \alpha_4
  \while \alpha_3 \le \alpha_1, \alpha_2 \le \alpha_4
\]
which is the inferred type of application $\fun{f} \val (\fun{x} f \, x)$, we can write
\[
  (\alpha \to[\drt] \beta) \to \alpha \to[\drt] \beta
\]
This simplification does incur some information loss --- see example at the beginning of Section~\ref{sec:unifying}.

\subsection{Displaying region parameters}
\label{sub:displaying-region-parameters}
Finally, we can get rid of region inequalities in the same way as dirt inequalities,
if we replace positive region parameters with their lower bounds.
The problem is that these lower bounds are more complex due to handled regions.
For inhabited parameters $\rgn*$, which have simpler constraints,
we can use the same technique, except that we also need to collect instance lower bounds.
For example, we may write a parameter $\rgn*$ with lower bounds
\[
  \inst_1 \in \rgn*, \inst_2 \in \rgn*, \rgn_1 \le \rgn*, \rgn_2 \le \rgn*
\]
as $\set{\inst_1, \inst_2} \cup \rgn_1 \cup \rgn_2$.

For other region parameters, we first employ the techniques suggested in Section~\ref{sub:further-simplifications}.
Then, we display each constraint $\rgn^- \le \rgn^+ \cup (\rgn*_1 \cupdot \cdots \cupdot \rgn*_n)$
by adding $\rgn^- \dotminus \rgn*_1 \dotminus \cdots \dotminus \rgn*_n$ to the lower bounds of $\rgn^+$.
We similarly add lower bounds of the form $\set{\inst} - \cdots$
for constraints involving instances.
We write $\dotminus$ to emphasise the fact that we remove a region only if it is a singleton.
If it turns out that any $\rgn*$ has a single lower bound $\inst \in \rgn*$,
we write $-\inst$ instead of $\dotminus \rgn*$.

If $\rgn \dotminus \cdots$ appears in multiple lower bounds but with different handled regions,
we display only the regions that appear in all constraints.
Thus, instead of
\[
  (\rgn \dotminus \rgn*_1 \dotminus \rgn*_2 \dotminus \rgn*_3) \cup
  (\rgn \dotminus \rgn*_2 \dotminus \rgn*_3 \dotminus \rgn*_4) \cup
  (\rgn \dotminus \rgn*_2 \dotminus \rgn*_3 \dotminus \rgn*_5)
\]
we write just $\rgn \dotminus \rgn*_2 \dotminus \rgn*_3$.

If there are multiple constraints with the same handled regions,
which happens when we use the same handler more than once,
we may merge them.
For example,
we may write
$\set{\inst} \cup (\set{\inst'} \cup \rgn) \dotminus \rgn*_1 \dotminus \rgn*_2$.
instead of
$\set{\inst} \cup (\set{\inst'} \dotminus \rgn*_1 \dotminus \rgn*_2) \cup (\rgn \dotminus \rgn*_1 \dotminus \rgn*_2)$

Still, this may be too much information in some cases, so we can also display $\rgn'$ as
$\set{\inst} \cup (\set{\inst'} \cup \rgn)?$
where the question mark lets the programmer know that some of instances in $\set{\inst'} \cup \rgn$ may be handled
and that this lower bound may be decreased if further information is available.
Finally, we can safely omit the question mark and over-approximate the lower bound with $\set{\inst, \inst'} \cup \rgn$.

Determining the exact level of detail to show is subjective,
and further techniques may arise when the effect system is used in practice.

\subsection{Displaying handler types}
As mentioned in Section~\ref{ssub:handling-exceptions},
handler types often have a repetitive form that we can write in a more compact way.
For example, a handler type
\[
  \alpha \E \set{\kord{lookup} \T \rgn_1, \kord{raise} \T \rgn_2 \mid \drt} \hto
  \beta \E \set{\kord{lookup} \T \rgn_1 \dotminus \rgn, \kord{raise} \T (\rgn_2 \dotminus \rgn') \cup \rgn'' \mid \drt}
\]
which describes a handler that handles memory lookups on a region $\rgn$,
handles exceptions from $\rgn'$ and raises exceptions from $\rgn''$,
can be written simply as
\[
  \alpha \hto[\kord{lookup} \T \dotminus\rgn,\ \kord{raise} \T \dotminus\rgn', +\rgn''] \beta
\]

\section{Examples}
\label{sec:examples}

\subsection{Function composition}

Before we turn to handlers,
we display a typical run of the algorithm on an example with no special algebraic features:
function composition, defined as
\[
  \kord{compose} \defeq \fun{f} \val (\fun{g} \val (\fun{x}
          \letin{y = f \, x}
          g \, y))
\]

The constraints are computed by the following derivation
\[
  \inferrule{
    \inferrule{
      \ctx \ent f \T \alpha_f \while \emptyset \\
      \ctx \ent x \T \alpha_x \while \emptyset
    }{
      \ctx \ent f \, x \T \alpha_1 \E \drt_1 \while \alpha_f \le (\alpha_x \to[\drt_1] \alpha_1)
    }
    \and
    \inferrule{
      \ctx, y\!:\!\alpha_y \ent g \T \alpha_g \\
      \ctx, y\!:\!\alpha_y \ent y \T \alpha_y
    }{
      \ctx, y\!:\!\alpha_y \ent g \, y \T \alpha_2 \E \drt_2 \while \alpha_g \le (\alpha_y \to[\drt_2] \alpha_2)
    }
  }{
    \inferrule{
      \inferrule{\ctx \ent \letin{y = f \, x} g \, y \T \alpha_2 \E \drt_3 \while \cstr}{\vdots}
    }{
      \emptyset \ent \kord{compose} \T \alpha_f \to[\drt_5] \alpha_g \to[\drt_4] \alpha_x \to[\drt_3] \alpha_2 \while \cstr
    }
  }
\]
where $\ctx = f \T \alpha_f, g \T \alpha_g, x \T \alpha_x$ and
\[
  \cstr = \set{
    \alpha_f \le (\alpha_x \to[\drt_1] \alpha_1),
    \alpha_g \le (\alpha_y \to[\drt_2] \alpha_2),
    \alpha_1 \le \alpha_y, \drt_1 \le \drt_3, \drt_2 \le \drt_3
  }
\]
Unifying the constraints, we see that $\alpha_f$ and $\alpha_g$ need to be replaced with fresh function types,
and the result of $\unify(\cstr)$ is the substitution
\[
  \sol = \set{
    \alpha_f \mapsto (\alpha_3 \to[\drt_6] \alpha_4),
    \alpha_g \mapsto (\alpha_5 \to[\drt_7] \alpha_6)  
  }
\]
and the unified constraints, which are:
\[
  \cstr' = \set{
    \alpha_x \le \alpha_3,
    \alpha_4 \le \alpha_1 \le \alpha_y \le \alpha_5,
    \alpha_6 \le \alpha_2,
    \drt_6 \le \drt_1 \le \drt_3,
    \drt_7 \le \drt_2 \le \drt_3  
  }
\]
Under $\sol$, the inferred type is
\[
  \kord{compose} \T (\alpha_3 \to[\drt_6] \alpha_4) \to[\drt_5] (\alpha_5 \to[\drt_7] \alpha_6) \to[\drt_4] (\alpha_x \to[\drt_3] \alpha_2) \]
so the sets of positive and negative parameters are
\[
  N = \set{\alpha_x, \alpha_4, \alpha_6, \drt_6, \drt_7}
  \qquad \text{and} \qquad
  P = \set{\alpha_2, \alpha_3, \alpha_5, \drt_3, \drt_4, \drt_5}
\]
After the garbage collection, the only constraints that remain are
\[
  \gc(\cstr') = \set{
    \alpha_x \le \alpha_3,
    \alpha_4 \le \alpha_5,
    \alpha_6 \le \alpha_2,
    \drt_6 \le \drt_3,
    \drt_7 \le \drt_3  
  }
\]
Finally, we replace all positive dirt parameters with the union of their lower bounds,
merge all type parameters in the same skeleton,
introduce fresh and readable parameters,
and obtain the final type
\[
  \kord{compose} \T
    (\alpha \to[\drt] \beta)
    \to (\beta \to[\drt'] \gamma)
    \to (\alpha \to[\drt \cup \drt'] \gamma)
\]

\subsection{Counting printouts}

Next, let us define a handler that computes the number of calls to $\kord{print}$ on a given channel $c$:
\begin{align*}
  \kord{count\_print} \defeq \fun{c} \val (&\handler \\
    &\case \val x \mapsto \val 0 \\
    &\case \call{c}{\kord{print}}{y}{k} \mapsto \letin{n = k \, \unt} \val (\succ n)\\
  &)
\end{align*}
The computed type of $\kord{count\_print}$ is $\alpha_c \to[\drt_0] (\C \hto \D)$,
where the form of the dirty types
\[
  \C \defeq \alpha_\text{in} \E \set{\kord{print} \T \rgn_\text{in} \mid \drt_\text{in}}
  \qquad\text{and}\qquad
  \D \defeq \alpha_\text{out} \E \set{\kord{print} \T \rgn_\text{out} \mid \drt_\text{out}}
\]
reflects that $\kord{print}$ is the only operation symbol, appearing in the handler,
while the computed constraints (in the order described on page~\pageref{pag:hand-rules}) are
\begin{multline*}
  \big\{
    (\unitty \to \D) \le (\unitty \to \alpha_2 \E \drt_2),
    \alpha_3 \le \natty,
    \alpha_2 \le \alpha_3,
    \drt_2 \le \drt_3,
    \drt_4 \le \drt_3, \\
    \natty \E \drt_1 \le \D,
    \natty \E \drt_3 \le \D,
    \alpha_c \le \kord{channel}^{\rgn*},
    \rgn_\text{in} \le \rgn_\text{out} \cupdot \rgn*,
    \drt_\text{in} \le \drt_\text{out}
  \big\}
\end{multline*}
After unification and garbage collection, we get that the type of $\kord{count\_print}$ is
\[
  \kord{channel}^{\rgn} \to[\drt_0]
  (\alpha_\text{in} \E \set{\kord{print} \T \rgn_\text{in} \mid \drt_\text{in}}
  \hto \natty \E \set{\kord{print} \T \rgn_\text{out} \mid \drt_\text{out}})
\]
under the constraints $\set{\drt_\text{in} \le \drt_\text{out}, \rgn \le \rgn*, \rgn_\text{in} \le \rgn_\text{out} \cupdot \rgn*}$.
If we rename the parameters and use notation introduced in Section~\ref{sec:displaying},
we may write the type as
\[
  \kord{channel}^{\rgn_1} \to (\alpha \E \set{\kord{print} \T \rgn_2 \mid \drt}
  \hto \natty \E \set{\kord{print} \T \rgn_2 \dotminus \rgn_1 \mid \drt})
\]
or even as $\kord{channel}^{\rgn} \to (\alpha \hto[\kord{print} \T \dotminus \rgn] \natty)$.

Define $c$ to be a simple imperative computation \inline{std#print "Hello, world!"} with the inferred type
$\unitty \E \set{\kord{print} \T \rgn \mid \drt}$ under the constraint $\kord{std} \in \rgn$.
If we use $\kord{count\_print}$ to count the number of $\hash{\kord{std}}{\kord{print}}$ calls as
\[
  \letin{h = (\kord{count\_print} \, \kord{std})} (\withhandle{h}{c})
\]
the inferred type of the handled computation is $\natty \E \set{\kord{print} \T \rgn \mid \drt}$
under the constraints $\set{\kord{std} \in \rgn*, \kord{std} \in \rgn \cupdot \rgn*}$.
Since $\rgn*$ will receive no further lower bounds, we may drop the second constraint
as described in Section~\ref{sub:further-simplifications},
thus both $\rgn$ and $\drt$ are completely unconstrained
and we may write the type as $\natty \E \emptyset$.

\section*{Conclusion}

\subsection*{Related work}

Our inference algorithm borrows heavily from existing inference algorithms:
type inference is based on one for structural subtyping of Simonet~\cite{simonet2003type},
region inference is based on one for non-structural subtyping by Pottier~\cite{pottier1998type},
and dirt inference is based on Remy's row-typing~\cite{remy1993type}.
Furthermore, our algorithm would be unusable in practice without garbage collection~\cite{pottier2001simplifying, simonet2003type, trifonov1996subtyping} ---
other lossless techniques presented in Section~\ref{sub:simplifying-region-constraints} merely complement it.
Our decision to store inferred types in one form but display them in the other is similar to one in~\cite{pottier1998type}
except that to simplify display, we discard not only invariants that were necessary for inference, but inferred information as well.
One point where we differ from current approaches is the use of unification for solving constraints.
Though unification is simple and familiar,
repeated substitution makes it very inefficient,
so in practice, it would be better to replace it with a constraint-based algorithm as described in~\cite{simonet2003type}.

The related effect systems can be roughly divided into three groups:
  effect systems with no handlers,
  effect systems for exception handlers,
  and effect systems for handlers of arbitrary algebraic effects.
So far, there are no approaches that offer general handlers and are not based on algebraic effects.

Most of the ongoing research considers effect systems for languages without any handlers.
One of the main aims in the early work was
a detailed analysis of memory allocation~\cite{lucassen1988polymorphic,talpin1992type,wadler1999marriage}.
Due to the lack of handlers, there is no interest in the exact locations being accessed,
only in the parts of the program that share locations from the same memory region.
Hence, ordinary unification together with a simple constraint resolution is enough to infer the information,
though our approach yields the same results.

Recent research acknowledges the importance of a user-friendly output (or input in prescriptive systems):
Scala~\cite{rytz2012lightweight}, a popular functional and object-oriented language,
has a prescriptive effect system where a programmer annotates functions with simple labels such as $\kord{@pure}$ or $\kord{@throws[IOException]}$.
Unfortunately, our effect system is descriptive, so we cannot compare the two at this point.

Koka~\cite{leijen2014koka}, a recently developed functional language,
is based on a descriptive effect system that represents inferred effects with a row of labels such as $\kord{exn}$, $\kord{io}$ or $\kord{div}$.
It is interesting to note that the effect system of Koka used to be based on constraints similar to ours~\cite{tate2010convenient},
but they were dropped in favour of rows as they were found to be too complex in practice.
We hope that our decision to keep constraints in the background alleviates this problem.

We can simulate the row-based approach to effect inference by assigning a single instance to each effect,
though this still gives us results with too precise dirt descriptions.
In order to obtain an output similar to Koka,
we need to merge all related dirt parameters into a single one,
just as we do for type parameters.

One effect that we do not treat, but Koka does, is divergence.
We can add a dummy operation $\hash{\star}{\kord{div}}$ and suitably modify the rule~\rulename{Cstr-LetRec} to track divergence,
but to prevent $\kord{div}$ from appearing in almost every practical program,
we need to augment the effect system with some form of termination analysis, just like Koka does.
For a practical language like Koka, this is essential,
though we leave it as future work in our development.

Effect systems for exceptions and their handlers~\cite{fahndrich1997program, leroy2000type, yi2002cost} often ignore other effects,
but provide much richer information about exception flow.
For example, Pessaux \& Leroy~\cite{leroy2000type} provide a row-based effect inference algorithm for OCaml
that uses control-flow analysis
in order to provide information about the values that exceptions carry as arguments.
This is because handlers in OCaml may be written so that they handle only exceptions with particular arguments,
for example only \inline{Failure "tl"}, which is raised when the tail function \inline{tl} is applied to an empty list.
We believe that in \Eff, declaring a new exception \inline{emptyListTail} is a cleaner solution
and one for which our approach already infers all important information.

Exception arguments aside, the inference algorithm of Pessaux \& Leroy produces results similar to ours
(for programs using only exceptions, that is).
We already saw at the beginning of Section~\ref{sec:eff} that both algorithms infer the same types of polymorphic higher-order functions.
To simulate the row-based approach to exception inference,
we could utilize our row-based dirt and take an \emph{operation symbol}~\inline{exc} with a single instance~$\star$ for each exception~\inline{exc}.
However, this prevents us from passing around exceptions as first-class values,
so it is better to represent each exception with a separate instance and obtain the same information in an unrestricted setting.

Handlers of algebraic effects are a recent discovery and so far,
there are only three effect systems beside ours that employ them.
First, Frank~\cite{mcbridefrank} is a prototype dependently-typed language with handlers and a prescriptive effect system.
Next, there is a library that provides a simpler form of handlers embedded as a domain specific language (DSL) inside a dependently-typed language Idris~\cite{brady2013programming}.
Finally, the closest to our approach is an effect system by Kammar, Lindley \& Oury~\cite{kammar2013handlers},
with safety results similar to ours,
and an implementation as a DSL inside Haskell.
This implementation also uses the type class mechanism of Haskell to infer effects.
The main contributions we bring to the group are first-class handlers and instances (other approaches use only operations to trigger effects),
and a stand-alone inference algorithm proven to be complete.

\subsection*{Future work}

Our inference algorithm is general,
offers strong guarantees on the effectful behaviour,
and presents the programmer with information that is easy to understand.
There are, as always, many possible improvements.

In its current form,
the presented inference algorithm infers the same types as the Hindley-Milner algorithm,
except that it provides an additional description of effects.
However, we can make the effect system \emph{prescriptive}
by adding constraints of the form $\rgn \le \Rgn$, which limit the allowed instances.
Then, a programmer may ensure that a given function will not raise exceptions
by simply ascribing it the type $\alpha \to[\set{\kord{raise} \T \emptyset \mid \drt}] \beta$.
Having such constraints allows us to lift the restrictions placed on the effect signature (discussed in Remark~\ref{rem:glitch}),
but determining their satisfiability is quite involved.
For example, $\rgn \le \emptyset$ may be satisfiable simply
because there are no constraints that give a lower bound to $\rgn$.
However, a more comprehensive treatment must also consider the case
when a $\rgn$ is empty because all of its instances have been removed by handlers,
and this is much more difficult to determine.

Next, let-polymorphism is currently based on value restriction, so that only types of expressions are generalized.
Since we already have an effect system in place,
we could relax this restriction and generalize the types of all pure computations,
but this again leads to the problem of determining empty regions as described in the above paragraph.

Furthermore, \Eff allows dynamic creation of fresh instances~\cite{bauer2012programming}.
Since instances are a crucial part of our type system,
instance generation has to be represented at that level as well.
One option is having dirty types of the form $\nu \Rgn. A \E \Drt$,
where $\Rgn$ captures the set of created instances, bound by $\nu$ in $A \E \Drt$.
Then, the type inference is particularly tricky as,
for example, the type of \inline{fun f -> f (); f ()} should reflect that it creates twice as much instances as \inline{f}.
A combination with recursion is a separate problem,
since programs can then create a potentially infinite number of instances,
though a wild card region $\top$ as discussed in Remark~\ref{rem:glitch} may be used in this case.

Finally, before algebraic effects and handlers can be considered practical,
we need an efficient way of evaluating programs that use them.
This may be difficult to achieve in general because of the freedom that handlers allow when manipulating the continuations.
At the very least, running any existing ML programs in the algebraic setting should induce a minimal overhead.

In this line, we could also pass any information inferred by the effect system
to an optimizing compiler:
pure computations may be postponed,
and computations with disjoint effects may be exchanged or even ran in parallel~\cite{tolmach1998optimizing}.
This looks like a promising and valuable direction of research,
especially because such optimizations have already been studied in the context of algebraic effects,
though without handlers~\cite{kammar2012algebraic}.

We hope that the presented work will help the ML community to recognize
algebraic effects as a natural progression of their current type system,
and that the Haskell community will acknowledge the additional flexibility that handlers have to offer.

\section*{Acknowledgements}

I would like to thank Andrej Bauer, Chris Stone, Ohad Kammar and the anonymous referees for all their extremely helpful comments and support.

\bibliographystyle{plain}
\bibliography{bibliography}

\appendix

\section{Proofs}
\label{app:proofs}

\begin{lem}
The following two rules are admissible:
\begin{mathpar}
  \inferrule[Sub-Refl]{
  }{
    A \le A
  }

  \inferrule[Sub-Trans]{
    A \le A' \\
    A' \le A''
  }{
    A \le A''
  }
\end{mathpar}
\end{lem}
\begin{proof}
  The proof proceeds by an induction on the structure of types.
\end{proof}

The subsumption rules \rulename{SubExpr} and \rulename{SubComp} allow us to
increase types of expressions and computations.
Conversely, we may decrease the types in contexts.

\begin{lem}
Take contexts $\ctx$ and $\ctx'$ that bind the same variables and assume that
for all $(x_i \T A_i) \in \ctx$, we have $(x_i \T A_i') \in \ctx'$ for some $A_i \le A_i'$.
Then, the following rules are admissible:
\begin{mathpar}
  \inferrule[SubCtxExpr]{
    \ctx' \ent e \T A
  }{
    \ctx \ent e \T A
  }

  \inferrule[SubCtxComp]{
    \ctx' \ent c \T \C
  }{
    \ctx \ent c \T \C
  }
\end{mathpar}
\end{lem}
\begin{proof}
  The proof proceeds by a routine induction on the derivation of the typing judgement.
\end{proof}

To relate monomorphic typing judgements of core \Eff to polymorphic inference judgements,
we first need to substitute away any variables from the polymorphic context $\pctx$.
For a polymorphic context $\pctx = (x_j \T \fra{F_j} A_j \while \cstr_j)_{j = 1}^n$
and expressions $(e_j)_{j = 1}^n$, we write $\ctx \ent (e_j)_j \models \pctx$
if for all $j = 1, \dots, n$, we have
$\ctx; (x_i \T \fra{F_i} A_i \while \cstr_i)_{i = 1}^{j - 1} \ent[F_j] e_j \T A_j \while \cstr_j$.
In this case, we define the expression $e[e_j / x_j]_j \defeq e[e_1 / x_1][e_2 / x_2] \cdots[e_n / x_n]$,
and the computation $c[e_j / x_j]_j \defeq c[e_1 / x_1][e_2 / x_2] \cdots[e_n / x_n]$.
Note that because polymorphic definitions may build on one another,
we need to use nested, and not simultaneous substitution.

\begin{prop}[Soundness]
\label{prop:soundness}
\hfill
\begin{itemize}
\item
  Assume that $\ctx; \pctx \ent[F] e \T A \while \cstr$ holds and take
  any expressions $(e_j)_j$ such that
  $\ctx \ent (e_j)_j \models \pctx$ holds
  and we have $\sol_j(\ctx) \ent e_j[e_i / x_j]_{i = 1}^{j - 1} \T \sol_j(A_j)$
  for any $\sol_j \models \cstr_j$.
  Then, for any solution $\sol \models \cstr$,
  we also have $\sol(\ctx) \ent e[e_j / x_j]_j \T \sol(A)$.
\item
  Assume that $\ctx; \pctx \ent[F] c \T \C \while \cstr$ holds and take
  any expressions $(e_j)_j$ such that
  $\ctx \ent (e_j)_j \models \pctx$ holds
  and we have $\sol_j(\ctx) \ent e_j[e_i / x_j]_{i = 1}^{j - 1} \T \sol_j(A_j)$
  for any $\sol_j \models \cstr_j$.
  Then, for any solution $\sol \models \cstr$,
  we also have $\sol(\ctx) \ent c[e_j / x_j]_j \T \sol(\C)$.
\end{itemize}
\end{prop}

\begin{proof}
  We proceed by a mutual induction on the derivation of inference judgements.
  In all cases, we use the same labels as in the considered rule.
  Also, as we never need the set of fresh parameters $F$, we do not write it.
  Take any suitable $(e_j)_j$ and any $\sol$ that satisfies the set of constraints in the conclusion,
  and consider the case when the last rule used in the derivation is:
  \begin{description}
    \item[\rulename{Cstr-PolyVar}]
      The case for a variable $x_i \in \pctx$
      is immediate because we can use the assumption made for $e_i = x_i[e_j / x_j]_j$.
    \item[\rulename{Cstr-Inst}]
      Since $\sol \models \inst \in \rgn*$,
      we have $\inst \in \sol(\rgn*)$.
      Because $\inst[e_j / x_j]_j = \inst$, we can use \rulename{Inst} to get
        $\sol(\ctx) \ent \inst[e_j / x_j]_j \T \sol(E^{\rgn*})$.
    \item[\rulename{Cstr-Hand}]
      Let $A \E \Drt = \sol(\C)$ and $B \E \Drt' = \sol(\D)$.
      By assumption, we have $\sol \models \cstr_v$, so we get $\sol(\ctx), x \T \sol(\alpha_\text{in}) \ent c_v[e_j / x_j]_j \T \sol(\D_v)$
      by the induction hypothesis.
      Since $\sol \models \D_v \le \D$, we further get $\sol(\ctx), x \T A \ent c_v[e_j / x_j]_j \T B \E \Drt'$.
      Similarly, we get appropriate results for each~$\Psi_i$.

      Next, for each $\op \in \ops$, let $\Rgn_\text{in}^\op = \sol(\rgn_\text{in}^\op)$ and $\Rgn_\text{out}^\op = \sol(\rgn_\text{out}^\op)$.
      Now, take any $\hash{\inst}{\op} \in \Drt$.
      Then, if $\op \in \ops$, we have $\inst \in \Rgn_\text{in}^\op$.
      Since $\sol \models \rgn_\text{in}^\op \le \rgn_\text{out}^\op \cup \uniq{\op = \op_i}{\rgn*_i}$,
      we either have $\inst \in \Rgn_\text{out}^\op$ hence $\hash{\inst}{\op} \in \Drt'$,
      or $\inst \in \uniq{\op = \op_i}{\Rgn_i}$.
      If $\op \not\in \ops$, then $\hash{\inst}{\op}$ must be in $\sol(\drt_{\text{in}})$,
      thus also in $\sol(\drt_\text{out}) \subseteq \Drt'$.
      We can then conclude by applying \rulename{Hand}.
    \item[\rulename{Cstr-Op}]
      As we have $\sol \models \cstr_1$, $\sol \models \cstr_2$, and $\sol \models \cstr$, we get
        $\sol(\ctx) \ent e_1[e_j / x_j]_j \T \sol(A_1)$,
        $\sol(\ctx) \ent e_2[e_j / x_j]_j \T \sol(A_2)$,
      and
        $\sol(\ctx), y \T B^\op \ent c[e_j / x_j]_j \T \sol(A \E \Drt)$
      by the induction hypothesis.

      Next, because $\sol \models A_1 \le E^{\rgn*}$ and $\sol \models A_2 \le A^\op$,
        we can use \rulename{SubExpr} to get
        $\sol(\ctx) \ent e_1[e_j / x_j]_j \T E^{\sol(\rgn*)}$
      and
        $\sol(\ctx) \ent e_2[e_j / x_j]_j \T A^\op$.
      Then, $\sol \models \Drt \le \set{\op \T \rgn \mid \drt}$,
      and by \rulename{SubComp}, we get
        $\sol(\ctx), y \T B^\op \ent c[e_j / x_j]_j \T \sol(A \E \set{\op \T \rgn \mid \drt})$.

      Finally, since $\sol(\rgn*) \subseteq \sol(\rgn)$, we have
        $\hash{\inst}{\op} \in \sol(\set{\op \T \rgn \mid \drt})$
      for any $\inst \in \sol(\rgn*)$, so we may use \rulename{Op} to conclude.
    \item[\rulename{Cstr-LetVal}]
      By the induction hypothesis for $e$, we have that the terms $e_1, \dots, e_n, e$ satisfy
      the conditions of the induction hypothesis for $c$,
      thus $\sol(\ctx) \ent (c[e / x])[e_j / x_j]_j \T \sol(\C)$ holds,
      and so, we have $\sol(\ctx) \ent (\letvalin{x = e} c)[e_j / x_j]_j \T \sol(\C)$ by \rulename{LetVal}.
    \item[\rulename{Cstr-With}]
      As $\sol \models \cstr_1$ and $\sol \models \cstr_2$,
      we can use induction to get
        $\sol(\ctx) \ent e[e_j / x_j]_j \T \sol(A)$
      and 
        $\sol(\ctx) \ent c[e_j / x_j]_j \T \sol(\C)$.
      Next, we have
        $\sol \models A \le (\C \hto \alpha \E \drt)$,
      thus we may use \rulename{SubExpr} and get
        $\sol(\ctx) \ent e[e_j / x_j]_j \T \sol(\C) \hto \sol(\alpha \E \drt)$.
      So by \rulename{With}, we get
        $\sol(\ctx) \ent (\withhandle{e}{c})[e_j / x_j]_j \T \sol(\alpha \E \drt)$.
  \end{description}
  In all other cases, the proof proceeds routinely.
\end{proof}

\begin{lem}[Weakening]
\label{lem:weakening}
\hfill
\begin{itemize}
\item
  If $\ctx; \pctx \ent[F] e \T A \while \cstr$ holds,
  so does $\ctx; \pctx, (x \T \fra{F'} A' \while \cstr') \ent[F] e \T A \while \cstr$
  for any $x$ that does not appear in $\ctx$, $\pctx$ or $e$.
\item
  If $\ctx; \pctx \ent[F] c \T \C \while \cstr$ holds,
  so does $\ctx; \pctx, (x \T \fra{F'} A' \while \cstr') \ent[F] c \T \C \while \cstr$
  for any $x$ that does not appear in $\ctx$, $\pctx$ or $c$.
\end{itemize}
\end{lem}
\begin{proof}
The proof proceeds by routine induction on the derivation of inference judgements.
\end{proof}

\begin{lem}[Exchange]
\label{lem:exchange}
\hfill
\begin{itemize}
\item
  If $\ctx; \pctx, (x_1 \T \fra{F_1} A_1 \while \cstr_1), (x_2 \T \fra{F_2} A_2 \while \cstr_2) \ent[F] e \T A \while \cstr$ holds,
  then so does $\ctx; \pctx, (x_2 \T \fra{F_2} A_2 \while \cstr_2), (x_1 \T \fra{F_1} A_1 \while \cstr_1) \ent[F] e \T A \while \cstr$.
\item
  If $\ctx; \pctx, (x_1 \T \fra{F_1} A_1 \while \cstr_1), (x_2 \T \fra{F_2} A_2 \while \cstr_2) \ent[F] c \T \C \while \cstr$ holds,
  then so does $\ctx; \pctx, (x_2 \T \fra{F_2} A_2 \while \cstr_2), (x_1 \T \fra{F_1} A_1 \while \cstr_1) \ent[F] c \T \C \while \cstr$.
\end{itemize}
\end{lem}
\begin{proof}
The proof proceeds by routine induction on the derivation of inference judgements.
\end{proof}

\begin{lem}
\label{lem:polyvar}
\hfill
\begin{itemize}
\item
  For any $e$, we have that
  if $\ctx; \pctx \ent[F] e[e' / x] \T A \while \cstr$ holds
  for some $\ctx; \pctx \ent[F'] e' \T A' \while \cstr'$,
  so does $\ctx; \pctx, (x \T \fra{F'} A' \while \cstr') \ent[F] e \T A \while \cstr$.
\item
  For any $c$, we have that
  if $\ctx; \pctx \ent[F] c[e' / x] \T \C \while \cstr$ holds
  for some $\ctx; \pctx \ent[F'] e' \T A' \while \cstr'$,
  so does $\ctx; \pctx, (x \T \fra{F'} A' \while \cstr') \ent[F] c \T \C \while \cstr$.
\end{itemize}
\end{lem}
\begin{proof}
  We proceed by a mutual induction on the structure of the term:
  \begin{itemize}
     \item If $e$ is some variable $y \ne x$, we have that $e[e' / x] = e$,
     hence $\ctx; \pctx \ent[F] e \T A \while \cstr$.
     We conclude by using Lemma~\ref{lem:weakening}.
     \item If $e$ is the variable $x$,
     we also have that $\ctx; \pctx \ent[F] e' \T A \while \cstr$ holds.
     By induction on the derivation of inference judgements,
     we can show that the inferred types and constraints are unique up to renaming.
     Thus $(\fra{F'} A' \while \cstr') = (\fra{F} A \while \cstr)$ up to $\alpha$-equivalence,
     so $\ctx; \pctx, (x \T \fra{F'} A' \while \cstr') \ent[F] x \T A \while \cstr$ holds.
     \item If $e$ is $\fun{x'} c$, the only way of obtaining the inference judgement
     is by using \rulename{Cstr-Fun} to get some
     $\ctx; \pctx \ent[F, \alpha] \fun{x} c[e' / x] \T \alpha \to \C \while \cstr$.
     Hence, we have that
     $\ctx, x \T \alpha; \pctx \ent[F] c[e' / x] \T \C \while \cstr$,
     so we get
     $\ctx, x \T \alpha; \pctx, (x \T \fra{F'} A' \while \cstr') \ent[F] c \T \C \while \cstr$
     by the induction hypothesis.
     Using \rulename{Cstr-Fun} we obtain the desired conclusion
     $\ctx; \pctx, (x \T \fra{F'} A' \while \cstr') \ent[F, \alpha] \fun{x} c \T \alpha \to \C \while \cstr$.
     \item If $e$ is $\letvalin{x' = e''} c$, the only rule that applies is \rulename{Cstr-LetVal}.
     We proceed just as in the previous case by using the induction hypothesis and reapplying the rule \rulename{Cstr-LetVal},
     except that we also need to use Lemma~\ref{lem:exchange} to obtain the proper ordering of variables in $\pctx$.
  \end{itemize}
  In all other cases that do not touch $\pctx$,
  the proof proceeds routinely just like for functions.
\end{proof}

\begin{prop}[Completeness]
\label{prop:completeness}
\hfill
\begin{itemize}
\item
  For any polymorphic context $\pctx$, closed substitution $\sol$, expressions $\sol(\ctx) \ent (e_j)_j \models \pctx$ and
  any $\sol(\ctx) \ent e[e_j / x_j]_j \T A$,
  we have $\pctx; \ctx \ent[F] e \T A' \while \cstr$ and
  there exists a solution $\sol' \models \cstr$, which extends $\sol$ to $F$,
  such that $\sol(A') \le A$.
\item
  For any polymorphic context $\pctx$, closed substitution $\sol$, expressions $\sol(\ctx) \ent (e_j)_j \models \pctx$ and
  any $\sol(\ctx) \ent c[e_j / x_j]_j \T \C$,
  we have $\pctx; \ctx \ent[F] c \T \C' \while \cstr$ and
  there exists a solution $\sol' \models \cstr$, which extends $\sol$ to $F$,
  such that $\sol(\C') \le \C$.
\end{itemize}
\end{prop}
\begin{proof}
  We again proceed by a mutual induction,
  this time on the derivation of typing judgements.
  In all cases, we again use the same labels as in the considered rule.

  Note that no parameters from $F$ can appear in $\ctx$,
  and we may safely assume that $\sol$ is undefined on $F$.
  Furthermore, we may safely compose substitutions that are defined on disjoint sets of parameters.
  For example, by taking $\sol_1$ and $\sol_2$ that
  extend $\sol$ to disjoint sets $F_1$ and $F_2$, respectively,
  we may uniquely define $\sol' = \sol \cup \sol_1 \cup \sol_2$,
  which extends $\sol$ to $F_1 \cup F_2$.

  Consider the case when the last rule used in the derivation is:
  \begin{description}
    \item[\rulename{Inst}]
      Assume that
        $\sol(\ctx) \ent \inst \T E^{\Rgn}$
      for some $\inst \in \Rgn$.
      By \rulename{Cstr-Inst}, we first get
        $\ctx \ent[\rgn*] \inst \T E^{\rgn*} \while \inst \in \rgn*$.
      Next, we extend $\sol$ to $\rgn*$ by defining $\sol' = \extend{\sol}{\rgn* \mapsto \Rgn}$.
      By assumption, we have $\inst \in \Rgn$,
      so $\sol' \models \inst \in \rgn*$ by definition.
      Finally, we get $\sol'(E^{\rgn*}) \le E^{\Rgn}$ by \rulename{Sub-Refl}.
    \item[\rulename{Hand}]
      From the value case and operation cases,
      we get the necessary premises $\prms_v$ and $(\prms_i)_i$ of \rulename{Cstr-Hand},
      and solutions $\sol_v$, $(\sol_i)_i$ and $(\sol_i')_i$ that satisfy the required properties.
      We then define
      \begin{align*}
        \sol = \sol_v \cup \smash{\bigcup_i \sol_i} \cup \smash{\bigcup_i \sol_i'} \cup \Big\{
        &\alpha_\text{in} \mapsto A,
        \alpha_\text{out} \mapsto B,
        \drt_\text{in} \mapsto \set{\hash{\inst}{\op} \in \Drt \mid \op \not\in \ops}, \\
        &\drt_\text{out} \mapsto \set{\hash{\inst}{\op} \in \Drt' \mid \op \not\in \ops},
        (\rgn*_i \mapsto \Rgn_i)_i, \\
        &(\rgn_\text{in}^\op \mapsto \set{\inst \mid \hash{\inst}{\op} \in \Drt})_{\op \in \ops}, \\
        &(\rgn_\text{out}^\op \mapsto \set{\inst \mid \hash{\inst}{\op} \in \Drt'})_{\op \in \ops}
        \Big\}
      \end{align*}
      We routinely check that $\sol$ satisfies all the necessary conditions
      and that we can apply \rulename{Cstr-Hand} to obtain the expected result.
    \item[\rulename{SubExpr}]
      By assumption, we have
        $\sol(\ctx) \ent e[e_j / x_j]_j \T A'$
      for some
        $\sol(\ctx) \ent e[e_j / x_j]_j \T A$ and $A \le A'$.
      By induction hypothesis, we have
        $\ctx \ent[F] e \T A'' \while \cstr$
      and $\sol' \models \cstr$ that extends $\sol$ to $F$
      such that $\sol'(A'') \le A$.
      We may then use \rulename{Sub-Trans} to get $\sol'(A'') \le A'$.
    \item[\rulename{Op}]
      Assume
        $\sol(\ctx) \ent (\call{e_1}{\op}{e_2}{\cont{y}{c}})[e_j / x_j]_j \T A \E \Drt$
      for some
        $\sol(\ctx) \ent e_1[e_j / x_j]_j \T E^R$,
        $\sol(\ctx) \ent e_2[e_j / x_j]_j \T A^\op$,
      and
        $\sol(\ctx), y \T B^\op \ent c[e_j / x_j]_j \T A \E \Drt$.
      By induction hypothesis, we get
        $\ctx \ent[F_1] e_1 \T A_1 \while \cstr_1$,
        $\ctx \ent[F_2] e_2 \T A_2 \while \cstr_2$,
      and
        $\ctx, y \T B^\op \ent[F] c \T \C \while \cstr$,
      together with $\sol_1$, $\sol_2$, and $\sol'$
      that satisfy the required properties.
      Using \rulename{Cstr-Op}, we get
      \begin{align*}
        \ctx &\ent[F_1, F_2, F, \rgn, \rgn*, \drt] \call{e_1}{\op}{e_2}{\cont{y}{c}} \T A \E \set{\op \T \rgn \mid \drt} \\
          &\while \cstr_1, \cstr_2, \cstr, A_1 \le E^{\rgn*}, A_2 \le A^\op, \rgn* \le \rgn, \Drt \le \set{\op \T \rgn \mid \drt}
      \end{align*}
      We extend $\sol$ to all the fresh parameters by
      \[
        \sol'' = \extend{\sol_1 \cup \sol_2 \cup \sol'}{\rgn \mapsto \Rgn, \rgn* \mapsto \Rgn, \drt \mapsto \set{\hash{\inst}{\op'} \in \Drt \mid \op' \ne \op}}
      \]
      We conclude the proof by observing that $\sol''$ satisfies all the necessary conditions.
    \item[\rulename{LetVal}]
      As $\sol(\ctx) \ent (\letvalin{x = e} c)[e_j / x_j]_j \T \C$,
      we get some $\sol(\ctx) \ent (c[e / x])[e_j / x_j]_j \T \C$.
      By induction hypothesis, we get some $\ctx; \pctx \ent[F_1] e \T A \while \cstr_1$
      and $\ctx; \pctx \ent[F_2] c[e / x] \T \C' \while \cstr_2$ together with $\sol \models \cstr_2$
      such that $\sol(\C') \le \C$.
      Then, by Lemma~\ref{lem:polyvar},
      we have that $\ctx; \pctx, (x \T \fra{F_1} A \while \cstr_1) \ent[F_2] c \T \C' \while \cstr_2$,
      so we may conclude by using \rulename{Cstr-LetVal}.
    \item[\rulename{With}]
      As
        $\sol(\ctx) \ent (\withhandle{e}{c})[e_j / x_j]_j \T A \E \Drt$
      we get
        $\sol(\ctx) \ent e[e_j / x_j]_j \T \C \hto A \E \Drt$
      and
        $\sol(\ctx) \ent c[e_j / x_j]_j \T \C$.
      Next, by induction hypothesis, we get
        $\ctx \ent[F_1] e \T A' \while \cstr_1$
      together with $\sol_1 \models \cstr_1$
      that extends $\sol$ to $F_1$,
      such that $\sol_1(A') \le \C \hto A \E \Drt$.
      Similarly, we get
        $\ctx \ent[F_2] c \T \C' \while \cstr_2$
      together with $\sol_2 \models \cstr_2$
      that extends $\sol$ to $F_2$,
      such that $\sol_2(\C') \le \C$.

      By \rulename{Cstr-With}, we get
      \[
        \ctx \ent[F_1, F_2, \alpha, \drt] \withhandle{e}{c} \T \alpha \E \drt
          \while \cstr_1, \cstr_2, A' \le (\C' \hto \alpha \E \drt)
      \]
      We define
      \[
        \sol' = \extend{\sol_1 \cup \sol_2}{\alpha \mapsto A, \drt \mapsto \Drt}
      \]
      which extends $\sol$ to all the fresh parameters.
      Since $\sol'$ extends $\sol_1$ and $\sol_2$,
      we have $\sol' \models \cstr_1$ and $\sol' \models \cstr_2$.
      Furthermore,
      \[
        \sol'(A')
        = \sol_1(A')
        \le \C \hto A \E \Drt
        \le \sol_2(\C') \hto A \E \Drt
        = \sol'(\C' \hto \alpha \E \drt) 
      \]
      hence $\sol'$ satisfies all the necessary conditions.
      Finally, we have $\sol'(\alpha \E \drt) \le A \E \Drt$ by \rulename{Sub-Refl}.
    \item[\rulename{SubComp}]
      The proof proceeds in the same way as in the case of \rulename{SubExpr}.
  \end{description}
  In all other cases, the proof again proceeds routinely.
\end{proof}

\begin{proof}[Proof of Theorem~\ref{thm:completeness}]
Take any $\ctx \ent e \T A'$.
We may then use Proposition~\ref{prop:completeness} for the empty context $\pctx$
and the identical substitution to obtain $\ctx \ent e \T A \while \cstr$
such that $A' \in \types{\cstr}$.
Since $A$ and $\cstr$ are uniquely determined up to a renaming,
we may use the same reasoning for any $A''$ and get $\set{ A'' \mid (\ctx \ent e \T A'') } \subseteq \types{\cstr}$.
The converse follows from Proposition~\ref{prop:soundness}
while the case for computations is exactly the same.
\end{proof}

\begin{proof}[Proof of Lemma~\ref{lem:unify}]
We can define a solution $\sol$ of a unified set of constraints~$\cstr$ as follows.
For all type parameters $\alpha$ and all dirt parameters $\drt$,
we define $\sol(\alpha) = \unitty$ and $\sol(\drt) = \emptyset$.
Then, $\sol$ trivially satisfies all type and dirt constraints in $\cstr$.
In a similar way, we also show that $\sol(\alpha) \approx \sol(\alpha')$ for any $\alpha \approx_\cstr \alpha'$.
In fact, we get a solution whenever we replace all type parameters in the same skeleton with the same type.

For region parameters, we define $\sol(\rgn) = \set{\inst \mid (\inst \in \rgn \cup \uniq{i}{\rgn*_i}) \in \cstr}$. 
Then, $\sol$ satisfies all constraints $(\inst \in \rgn \cup \uniq{i}{\rgn*_i}) \in \cstr$.
Next, take a constraint $(\rgn \le \rgn' \cup \uniq{i \in I}{\rgn*_i})$.
Then $(\inst \in \rgn \cup \uniq{i \in J}{\rgn*_i}) \in \cstr$ implies
$(\inst \le \rgn' \cup \uniq{i \in I \cup J}{\rgn*_i}) \in \cstr$ because $\cstr$ is unified.
Thus, if $\inst \in \sol(\rgn)$ then $\inst \in \sol(\rgn')$,
and so $\sol \models \rgn \le \rgn' \cup \uniq{i}{\rgn*_i}$.

There are, of course, countless other ways of choosing $\sol$.
\end{proof}

\begin{proof}[Proof of Proposition~\ref{prop:unify}]

We see that $\unify$ always terminates because at each recursive call,
we strictly decrease the degree, defined lexicographically according to:
\begin{enumerate}
\item
  the number of skeletons of $\approx_\cstr$,
\item
  the number of all type constructors in $\queue$,
\item
  the total \emph{potential} of all dirt parameters,
  where the potential of a dirt parameter~$\drt^\ops$ is the number
  of all operation symbols mentioned in $\queue$ that do not appear in $\ops$,
\item
  the number of constraints in $\queue$.
\end{enumerate}
Because of occur check, (1) decreases each time we unify any constraint of the form $\alpha \le A$ or $A \le \alpha$.
When decomposing constraints between types of the same form, we keep (1) as it is, but decrease (2).
When we unify a constraint $\alpha \le \alpha'$, (1) decreases if $\alpha$ was not in the same skeleton as $\alpha'$.
If it was, (1--3) remain the same, but (4) decreases.
A similar situation occurs when unifying region constraints or dirt constraints with matching operations.
If we unify dirt with non-matching annotations, we keep type constraints and thus (1--2) as they are,
but ensure that (3) decreases.

To show that unification preserves solutions, we observe that at each recursive call
$\unify(\sol_k; \cstr_k; \queue_k) = \unify(\sol_{k + 1}; \cstr_{k + 1}; \queue_{k + 1})$,
and for each $\sol_k' \models \cstr_k \cup \queue_k$,
there exists $\sol_{k + 1}' \models \cstr_{k + 1} \cup \queue_{k + 1}$
such that $\sol_k' \circ \sol_k = \sol_{k + 1}' \circ \sol_{k + 1}$.
Similarly, we observe that if $\unify(\sol; \cstr; \queue)$ results in a failure,
then $\cstr \cup \queue$ does not have any solutions.

Since we start with $\sol_0 = \mathrm{id}$, $\cstr_0 = \emptyset$ and $\queue_0 = \cstr$
and end with some $\sol_n = \sol'$, $\cstr_n = \cstr'$ and $\queue_n = \emptyset$,
then for each $\sol \models \cstr$, we have some $\sol'' \models \cstr'$ such that $\sol = \sol'' \circ \sol'$.
\end{proof}

\begin{proof}[Proof of Corollary~\ref{cor:unify}]
Let us show the soundness of rule \rulename{Unify-Expr}, because the proof for \rulename{Unify-Comp} is exactly the same.
First, take any $\sol' \models \cstr'$.
By Proposition~\ref{prop:unify}, we know that $\sol' \circ \sol \models \cstr$,
so by the induction hypothesis, we get $(\sol' \circ \sol)(\ctx) \ent e \T (\sol' \circ \sol)(A)$,
which is exactly what we wanted to prove.

For the other direction, take any $\sol'' \models \cstr$.
Again by Proposition~\ref{prop:unify}, there must exist some $\sol' \models \cstr'$
such that $\sol'' = \sol' \circ \sol$.
We then get $\sol'(\sol(\ctx)) \ent e \T \sol'(\sol(A))$ or, equivalently,
$\sol''(\ctx) \ent e \T \sol''(A)$ by the induction hypothesis.
\end{proof}

\begin{proof}[Proof of Proposition~\ref{prop:garbage-collection}]
Again, we shall prove only the soundness of rule \rulename{GC-Expr}, because the proof for \rulename{GC-Comp} is identical.
The soundness of the reverse rule is trivial:
any solution $\sol \models \cstr$ is also a solution of $\gc(\cstr)$ so by induction hypothesis, we get $\sol(\ctx) \ent e \T \sol(A)$.

For the forward direction, the reasoning follows the one sketched in the beginning of Section~\ref{sub:garbage-collection}.
Take any solution $\sol' \models \gc(\cstr)$.
Following~\cite{simonet2003type},
we are going to construct a solution $\sol \models \cstr$
such that $\sol(\alpha^+) \le \sol'(\alpha^+)$ holds for all $\alpha^+ \in P$
and $\sol'(\alpha^-) \le \sol(\alpha^-)$ holds for all $\alpha^- \in N$
(and similarly for region and dirt parameters).
Then, by the induction hypothesis, we get $\sol(\ctx) \ent e \T \sol(A)$.
Since $P$ contains $\pos(A)$ and $N$ contains $\neg(A)$, we have $\sol(A) \le \sol'(A)$,
thus by \rulename{SubExpr} we get $\sol(\ctx) \ent e \T \sol'(A)$.
Conversely, for all $(x_i \T A_i) \in \ctx$, the set $P$ contains $\neg(A_i)$ and $N$ contains $\pos(A_i)$,
so we get $\sol'(A_i) \le \sol(A_i)$.
We may thus use \rulename{SubCtxExpr} and get $\sol'(\ctx) \ent e \T \sol'(A)$.

We construct $\sol \models \cstr$ as follows.
Let us start with the simplest case of dirt parameters.
We define
\newcommand{\mini}[2][3]{\hspace{-#1em}#2\hspace{-#1em}}
\newcommand{\ministack}[2][3]{\mini[#1]{\substack{#2}}}
\[
  \sol(\drt) \defeq \bigcup_{\ministack{(\drt^- \le \drt) \in \cstr \\ \drt^- \in N}} \sol'(\drt^-)
\]
First, for any $\drt^- \in N$,
we implicitly have $(\drt^- \le \drt^-) \in \cstr$,
so $\sol'(\drt^-) \subseteq \sol(\drt^-)$.
Next, for any $\drt^+ \in P$, if have $(\drt^- \le \drt^+) \in \cstr$,
we also have $\drt^- \le \drt^+ \in \gc(\cstr)$ since $\drt^- \in N$.
Because $\sol' \models \gc(\cstr)$, we have $\sol'(\drt^-) \subseteq \sol'(\drt^+)$,
so
\[
  \sol(\drt^+) = \bigcup_{\drt^- \le \drt^+} \sol'(\drt^-) \subseteq \bigcup_{\drt^- \le \drt^+} \sol'(\drt^+) = \sol'(\drt^+)
\]

Finally, we need to show that $\sol$ satisfies all dirt constraints in $\cstr$.
Take some $(\drt \le \drt') \in \cstr$.
Now, if $(\drt^- \le \drt) \in \cstr$ holds for some $\drt^- \in N$,
then $(\drt^- \le \drt') \in \cstr$ holds as well
because $\cstr$ is unified and therefore closed under logical implication.
Thus $\sol(\drt')$ is defined as a union over a bigger index set than $\sol(\drt)$,
so $\sol(\drt) \subseteq \sol(\drt')$ and $\sol \models \drt \le \drt'$.

For type parameters~$\alpha$, we proceed similarly and define
\[
  \sol(\alpha) \defeq \bigcup_{\ministack{(\alpha^- \le \alpha) \in \cstr \\ \alpha^- \in N}} \sol'(\alpha^-)
\]
where we define union over closed types as:
\begin{align*}
  \boolty \cup \boolty &= \boolty \qquad&
  (A_1 \to \C_1) \cup (A_2 \to \C_2) &= (A_1 \cap A_2) \to (\C_1 \cup \C_2) \\
  \natty \cup \natty &= \natty &
  E^{\Rgn_1} \cup E^{\Rgn_2} &= E^{\Rgn_1 \cup \Rgn_2} \\
  \unitty \cup \unitty &= \unitty &
  (\C_1 \hto \D_1) \cup (\C_2 \hto \D_2) &= (\C_1 \cap \C_2) \hto (\D_1 \cup \D_2) \\
  \emptyty \cup \emptyty &= \emptyty &
  (A_1 \E \Drt_1) \cup (A_2 \E \Drt_2) &= (A_1 \cup A_2) \E (\Drt_1 \cup \Drt_2)
\end{align*}
The intersection is defined dually.
We can see that the union is well defined
because for all $(\alpha^- \le \alpha) \in \cstr$, we have $\alpha^- \approx_\cstr \alpha$,
hence all such $\alpha^-$ belong to the same skeleton of $\approx_\cstr$, and so also of $\approx_{\gc(\cstr)}$.
Since $\sol'$ is a solution of $\gc(\cstr)$, all types $\sol'(\alpha^-)$ are of the same shape.

We also need to consider the case when the union is empty because
there are no $\alpha^- \in N$ such that $(\alpha^- \le \alpha) \in \cstr$.
In this case, we set $\sol(\alpha)$ to be the intersection of $\sol'(\alpha')$ for all $\alpha \approx_\cstr \alpha'$.
Then, $\sol(\alpha) \le \sol'(\alpha)$ holds in case $\alpha \in P$, and $\sol$ satisfies all constraints $(\alpha \le \alpha') \in \cstr$.
If we have a constraint $(\alpha' \le \alpha) \in \cstr$, then $\alpha'$ also cannot have any negative lower bounds because $\cstr$ is closed, hence $\sol(\alpha') = \sol(\alpha)$.
Finally, the case $\alpha \in N$ cannot occur.

For region parameters~$\rgn$, we need to take both instances and handled regions into an account.
So, we define
\[
  \sol(\rgn) \defeq \Big(\bigcup_{\ministack[0.6]{(\rgn^- \le \rgn \cup \uniq{i}{\rgn*_i}) \in \cstr \\ \rgn^- \in N}} \big(\sol'(\rgn^-) - \uniq{i}{\sol'(\rgn*_i)}\big)\Big) \cup \Big(\bigcup_{{(\inst \in \rgn \cup \uniq{i}{\rgn*_i}) \in \cstr}} \big(\set{\inst} - \uniq{i}{\sol'(\rgn*_i)}\big)\Big)
\]
As before, we get $\sol'(\rgn^-) \le \sol(\rgn^-)$ for all $\rgn^- \in N$ and $\sol(\rgn^+) \le \sol'(\rgn^+)$ for all $\rgn^+ \in P$.
Next, let us show that $\sol$ satisfies all region constraints in $\cstr$.
Let us consider only ones of the form $(\rgn \le \rgn' \cup \uniq{i \in I}{\rgn*_i}) \in \cstr$
because the proof for ones with instances is similar.

For any $(\rgn^- \le \rgn \cup \uniq{i \in J}{\rgn*_i}) \in \cstr$
contributing to $\sol(\rgn)$, we have
$(\rgn^- \le \rgn' \cup \uniq{i \in I \cup J}{\rgn*_i}) \in \cstr$ because $\cstr$ is unified.
Thus, $\sol'(\rgn^-) - \uniq{i \in I \cup J}{\sol'(\rgn*_i)} \subseteq \sol(\rgn')$
by the definition of $\sol(\rgn')$.
If we add $\uniq{i \in I}{\sol'(\rgn*_i)}$ to both sides, we get
\[
  \big(\sol'(\rgn^-) - \uniq{i \in J}{\sol'(\rgn*_i)}\big) \subseteq \sol(\rgn') \cup \uniq{i \in I}{\sol'(\rgn*_i)}
\]
We similarly get
\[
  \big(\set{\inst} - \uniq{i \in J}{\sol'(\rgn*_i)}\big) \subseteq \sol(\rgn') \cup \uniq{i \in I}{\sol'(\rgn*_i)}
\]
for each $(\inst \in \rgn \cup \uniq{i \in J}{\rgn*_i}) \in \cstr$ that contributes to $\sol(\rgn)$.
These are all contributions to $\sol(\rgn)$, thus
\[
  \sol(\rgn) \subseteq \sol(\rgn') \cup \uniq{i \in I}{\sol'(\rgn*_i)}
\]
Now comes the critical step: since $\rgn*_i \in P$, we have $\sol(\rgn*_i) \subseteq \sol'(\rgn*_i)$.
However, $\sol(\rgn*_i)$ was non-empty and if we increase it, it contributes less to the singleton union,
therefore
\[
  \sol(\rgn) \subseteq \sol(\rgn') \cup \uniq{i \in I}{\sol(\rgn*_i)}
\]
and $\sol \models (\rgn \le \rgn' \cup \uniq{i \in I}{\rgn*_i})$.

A careful reader might have observed that not all parameters $\rgn*_i$ occur in $P$,
only those that appear in constraints $\rgn^- \le \rgn^+ \cup \uniq{i}{\rgn*_i}$
where $\rgn^- \in N$ and $\rgn^+ \in P_0$ (and similar ones for instances).
However, this is easy to fix.

First, take $P'$ to be $P_0$ and the set of all parameters $\rgn*_i$ that appear in \emph{any} singleton union in \emph{any} constraint in $\cstr$.
In this case, the above reasoning is valid and the set of constraints $\gc[P', N](\cstr)$ is equivalent to $\cstr$.
Now, repeat the whole process and take $P''$ to be $P_0$ and all parameters that appear in any singleton union in any constraint in $\gc[P', N](\cstr)$.
Again, $\gc[P'', N](\cstr)$ is equivalent to $\gc[P', N](\cstr)$ and so also to $\cstr$.
However, the only region constraints left in $\gc[P', N](\cstr)$ are ones of the form $\rgn^- \le \rgn^+ \cup \uniq{i}{\rgn*_i}$
with $\rgn^- \in N$ and $\rgn^+ \in P_0$, thus $P'' = P$.

In the end, let us show that $\gc(\cstr)$ is unified if $\cstr$ is.
We can consider only the case for type constraints as for other cases the proof is almost exactly the same.
Take constraints $\alpha_1 \le \alpha_2$ and $\alpha_2 \le \alpha_3$ in $\gc(\cstr)$.
Since $\gc(\cstr) \subseteq \cstr$, these two constraints must also be in $\cstr$,
which is unified, hence $\alpha_1 \le \alpha_3$ is in $\cstr$ as well.
From our assumption we have $\alpha_1, \alpha_2 \in N$ and $\alpha_2, \alpha_3 \in P$,
therefore $(\alpha_1 \le \alpha_3) \in \gc(\cstr)$.
We can similarly show that if we have $\alpha \le \alpha' \in \gc(\cstr)$,
then $\alpha \approx_{\gc(\cstr)} \alpha'$.
\end{proof}

\end{document}